\documentclass[10pt,journal,compsoc]{IEEEtran}

\ifCLASSOPTIONcompsoc
  \usepackage[nocompress]{cite}
\else
  \usepackage{cite}
\fi
\ifCLASSINFOpdf

\else
\fi

\usepackage{orcidlink}
\usepackage{url}
\usepackage{capt-of}  
\usepackage{cuted}  
\usepackage{enumitem}
\usepackage{amsmath, amsthm, mathtools, amssymb}
\usepackage{xcolor,colortbl}
\usepackage{diagbox}
\hypersetup{hidelinks}
\usepackage[markup=underlined]{changes}

\newtheorem{theorem}{Theorem}
\newtheorem{example}{Example}
\newtheorem{lemma}{Lemma}

\ifCLASSOPTIONcompsoc
  \usepackage[nocompress]{cite}
\else
  \usepackage{cite}
\fi

\ifCLASSINFOpdf
  \usepackage{graphicx}
  \DeclareGraphicsExtensions{.pdf,.jpeg,.png}
\else
\fi

\usepackage{amsmath}
\usepackage{algorithm}
\usepackage{algpseudocode}
\algnewcommand\algorithmicforeach{\textbf{for each}}
\algdef{S}[FOR]{ForEach}[1]{\algorithmicforeach\ #1\ \algorithmicdo}

\ifCLASSOPTIONcompsoc
 \usepackage[caption=false,font=footnotesize,labelfont=sf,textfont=sf]{subfig}
\else
 \usepackage[caption=false,font=footnotesize]{subfig}
\fi

\usepackage{xr}
\makeatletter

\newcommand*{\addFileDependency}[1]{
\@addtofilelist{#1}

\IfFileExists{#1}{}{\typeout{No file #1.}}
}\makeatother

\begin{document}

\title{Generating Temporal Contact Graphs \\ Using Random Walkers}

\author{Anton-David~Almasan\orcidlink{0009-0009-1583-6840},
        Sergey Shvydun\orcidlink{0000-0002-6031-8614},
        Ingo Scholtes\orcidlink{0000-0003-2253-0216},
        and Piet Van Mieghem\orcidlink{0000-0002-3786-7922},~\IEEEmembership{Fellow,~IEEE}%
\IEEEcompsocitemizethanks{
\IEEEcompsocthanksitem David Almasan, Sergey Shvydun and Piet Van Mieghem are with the Faculty of Electrical Engineering, Mathematics and Computer Science, Delft University of Technology, 2628 CD Delft, The Netherlands.
\IEEEcompsocthanksitem Ingo Scholtes is with the Chair of Machine Learning for Complex Networks at the Center for Artificial Intelligence and Data Science (CAIDAS), Julius-Maximilians-Universität Würzburg, Emil-Fischer-Strasse 50
D-97074 Würzburg,
Germany.

\IEEEcompsocthanksitem (e-mail: A.D.Almasan@tudelft.nl; S.Shvydun@tudelft.nl; ingo.scholtes@uni-wuerzburg.de; P.F.A.VanMieghem@tudelft.nl).}%
}

\IEEEtitleabstractindextext{%
\begin{abstract}
We study human mobility networks through timeseries of contacts between individuals. Our proposed Random Walkers Induced temporal Graph (RWIG) model generates temporal graph sequences based on independent random walkers that traverse an underlying graph in discrete time steps. Co-location of walkers at a given node and time defines an individual-level contact. RWIG is shown to be a realistic model for temporal human contact graphs, which may place RWIG on a same footing as the Erdos-Renyi (ER) and Barabasi-Albert (BA) models for fixed graphs. Moreover, RWIG is analytically feasible: we derive closed form solutions for the probability distribution of contact graphs.
\end{abstract}
}

\maketitle

\IEEEdisplaynontitleabstractindextext
\IEEEpeerreviewmaketitle

\IEEEraisesectionheading{\section{Introduction}\label{sec:introduction}}

\IEEEPARstart{I}{n} the past years, the study of temporal graphs has received a surge of interest, e.g. to model how time-varying human contact patterns impact epidemics like COVID-19 \cite{covid1,covid2,covid3}.
Empirical studies of real-world contact patterns have identified several characteristics of temporal graphs that can influence dynamical processes.

\begin{figure*}[!t]
\centering
\subfloat{\includegraphics[width=2.5in]{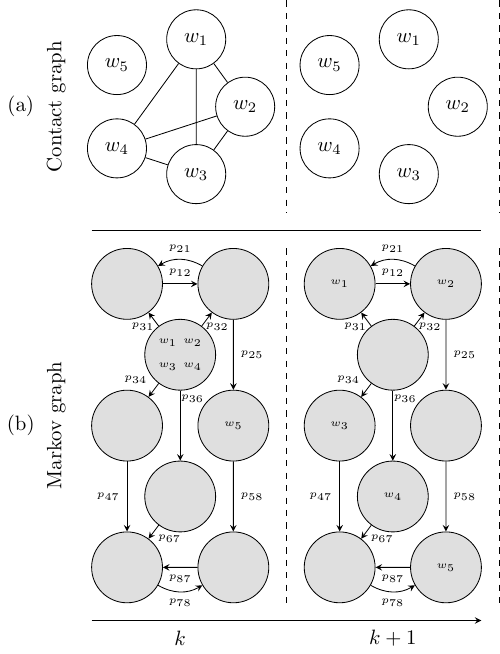}}%
\subfloat{\includegraphics[width=2.5in]{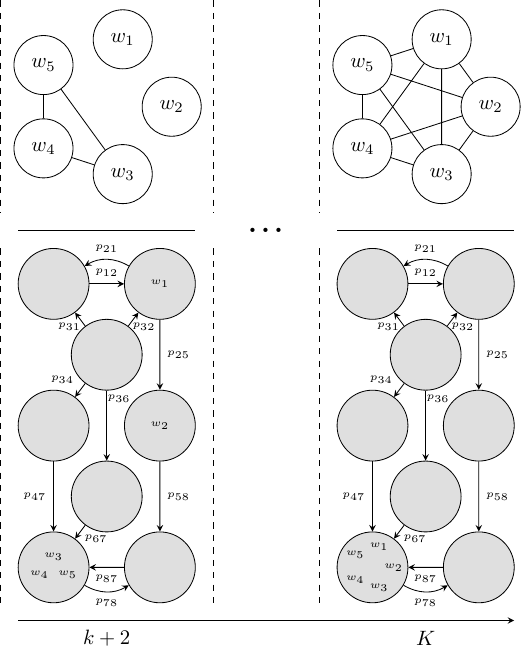}}%
\vspace{-0.25cm}\caption{RWIG. (a) Contact graphs. (b) Random walkers traversing the Markov graph.}
\label{fig:g_t realisation}
\end{figure*}

A first line of research has focused on the question how the \emph{temporal distribution of interactions} affects the evolution of dynamical processes in temporal graphs.
Studies on the influence of non-Poissonian and bursty node activity patterns \cite{Iribarren2009,Karsai2013,Takaguchi2013,Perotti2014} or long-lasting or concurrent interactions \cite{Masuda2014,Morris1995} have shown that real contact patterns can both slow down or speed up spreading processes compared to a static graph, where all links are always active.

A second line of research has addressed the question \emph{how the temporal ordering of interactions} influences dynamical processes such as diffusion or epidemic spreading.
In a nutshell, for two temporal contacts occurring between Alice and Bob at time $t$ and between Bob and Carol at time $t'$, a virus can only spread from Alice via Bob to Carol if the contact between Alice and Bob occurs \emph{before} the contact between Bob and Carol, i.e. iff $t{<}t'$.
If the temporal ordering of contacts is reversed, no \emph{time-respecting path} exists between Alice and Carol due to the directedness of the \emph{arrow of time}.
Empirical studies on social, biological, and technical systems \cite{Lentz2013,Pan2011,Holme2012} have shown that the \emph{causal topology} of temporal graphs, i.e. who can influence whom via time-respecting paths, is more complex than what we expect from their static, time-aggregated counterparts, leading to non-trivial effects such as a speed up or slow down of diffusion processes compared to (randomized) temporal graphs, that lack correlations in the temporal ordering of interactions \cite{Pfitzner2013,Scholtes2014,Rosvall2014,Lentz2016}.

Several temporal graph modelling and learning approaches have been proposed that account for some of the complex characteristics of empirical contact patterns 
\cite{Bois2015,Zhou2020,Zeno2021,Longa2024,Scholtes2017,Lambiotte2019}. 
Recently, a system theoretical approach towards emulating temporal graphs is presented in \cite{Shvydun_Van Mieghem2024}. 
Various approaches to model human mobility, which could explain some of the temporal characteristics of human contact patterns, are discussed in \cite{Barbosa2018,Chang2022,Panisson2012,Mauro2022}. Beside epidemic spreading, a better understanding of temporal mechanisms can also facilitate the design, management and control of mobile opportunistic networks \cite{Kui2018} or human mobility in public transportation networks \cite{Huyn2022}.
However, we still lack simple \emph{generative models for temporal graphs} that (i) are able to reproduce realistic contact patterns, (ii) facilitate analytic treatment and (iii) shed light on potential mechanisms that shape both the topological and temporal dimension of temporal graphs.

Addressing this research gap, we propose the \textit{Random Walkers Induced temporal Graph} (\textit{RWIG}) model, which uses multiple random walkers on a finite graph as a generative model for temporal contact networks. 
Any realization of a discrete-time Markov process on $N$ states can be represented by a random walk on the corresponding Markov graph with $N$ nodes (states), where a link between two states $i$ and $j$ is characterised, i.e. both directed and weighted, by the transition probability $p_{ij}$. 
The RWIG model considers a collection of $M$ random walkers that simultaneously traverse the Markov graph in discrete timesteps according to the $N{\times}N$ Markov transition probability matrix $P$ with elements equal to the transition probabilities $p_{ij}$. Hence, each walker executes a realization of the same Markov process or, equivalently, each walker's trajectory is driven by the Markov process. 
Thus, we assume in RWIG that the Markov process generates human mobility trajectories over a set of places (states). 
Next to the Markov graph, at discrete time $k$, the contact graph $G_k$ with $M$ nodes is generated, in which the nodes represent the random walkers. The main assumption of RWIG is that \textit{links in the contact graph $G_k$ are created between walkers which visit the same state in the Markov graph at discrete time $k$}.
Figure \ref{fig:g_t realisation} exemplifies $M{=}5$ random walkers, who traverse a Markov graph with $N{=}8$ states (shaded) in discrete-time steps according to the transition probabilities $p_{ij}$. The probabilities $p_{ij}$ are depicted in Figure \ref{fig:g_t realisation}(b) on the links between Markov states. The observation window has length $K$ as displayed on the horizontal axis and four discrete time steps are shown.  In Figure \ref{fig:g_t realisation}(a), RWIG generates the contact graph of the 5 walkers at each timestep by creating links between all walkers found in the same state in the Markov graph. For instance, at time $k{+}2$, walkers $w_3, w_4$ and $w_5$ are in the same state and thus form a fully connected subgraph or clique in the contact graph $G_{k+2}$, separated from the single node cliques of walker $w_1$ and $w_2$. 

A physical interpretation of RWIG is a collection of individuals moving through space. The underlying graph with adjacency matrix \(A\) represents a city map, with nodes as various locations (e.g. restaurants, workplaces, homes, public transport stations, etc) and links as physical paths between locations. The random walkers represent individuals and the transition probabilities $p_{ij}$ assume that all individuals behave the same.

We can regard the probabilities $p_{ij}$, which together form an $N{\times}N$ transition probability matrix $P$, as a common \textit{policy}, which all individuals follow. The transition probability matrix $P$ can generally take the form of any function $f(A)$ of the adjacency matrix $A$ which results in a stochastic matrix \cite{PVM_graphspectra_second_edition}. An example in which the probabilities of jumping from a state $i$ to any other adjacent state $j$ are all equal is $P=\Delta^{-1}A$, where $\Delta$ is the diagonal matrix of the degree vector of the underlying graph with adjacency matrix $A$.

As a common policy is restrictive and often unrealistic (e.g. a kindergartener would visit different locations than an office worker), we consider that each random walker $w_r$ can have a different policy, or transition probability matrix $P_r$. All policies, however, still reflect the same underlying graph topology (e.g. city map). Consequently, if there is no link between two states $i$ and $j$ in the adjacency matrix $A$, then all policies must have a zero probability for state transitions between nodes $i$ and $j$ (i.e. $a_{ij}{=}0$ implies that $(P_r)_{ij}{=}0$, for all integers $i, j \in \{1, ..., N\}$).

Although the properties of random walks have been extensively studied, the dynamics of multiple random walks on a graph still represents an active research area. Riascos and Sanders \cite{meanencountertimes} study multiple non-interactive random walkers on a graph and analyse the mean encounter times of walkers. 
A similar model is proposed to generate contacts between individuals in \cite{markovian random walk model of epidemic spreading}, which are then used to study the evolution of epidemics. Masuda \textit{et al}. \cite{comprehensiveRWstudy} present a detailed study of the theory and applications of random walks. To the best of our knowledge, RWIG is the first model which leverages multiple random walks to generate temporal graphs.
Our contribution can be summarised:
\begin{itemize}[noitemsep,topsep=0pt]
 \item We propose the RWIG model based on random walkers for generating temporal contact networks.
  \item We provide an analytical formula for the probability distribution of the contact graphs, which are produced by RWIG given the transition matrices $\{P_r\}_{r=1}^M$ and the initial states of all walkers.
  \item We demonstrate how RWIG is able to generate contact graphs that resemble real temporal networks.
\end{itemize}
The paper is organised as follows. In Section \ref{section: G_t enum main section}, we describe the state space and topological structure of contact graphs. Section \ref{section: main section} provides an analytical formula for the probability distribution of the contact graph formed by a set of walkers, conditioned on the walkers' initial states and policies. Section \ref{main section: steady-state graphs} discusses RWIG in the steady-state. To motivate the applicability of RWIG, Section \ref{section: data} offers simulation results illustrating the wide variety of contact graphs produced by RWIG and compares the RWIG generated sequences with empirical data. Finally, we introduce the notation to the reader in Appendix \ref{apx: symbols} and mathematical definitions are deferred to Appendix \ref{apx:combinatorics}.

\section{Random Walkers Induced temporal Graph (RWIG)}\label{section: G_t enum main section}
\subsection{Formulation of RWIG}
We consider an undirected unweighted graph with $N$ nodes and $L$ links that is represented by an $N{\times}N$ adjacency matrix $A$, which is the underlying graph. A Markov graph that emulates a random walk on that graph has the $N{\times}N$ probability transition matrix $P$. For instance, the transition probability matrix $P{=}\Delta^{-1}A$, where $\Delta{=}\left(d_{1},d_{2},{\ldots},d_{N}\right)  $ and $d_{i}$ is the degree of node $i$, describes a Markov graph \cite[p. 108-110]{PVM_graphspectra_second_edition} in which there is an equal probability to reach neighbouring states. On that Markov graph, $M$ random walkers, independently of each other, jump from one state to another state per discrete time $k$, starting from $k{=}0$ until some finite time $k{=}K$, according to the $N{\times}N$ probability transition matrix $P$. The trajectory of each random walker $w_j{\in}\left\{w_1,{\ldots},w_M\right\}$ across the states of the Markov graph can be regarded as one realization of the Markov process \cite{perfanalysis}, that starts in the state described by the \(1{\times}N\) vector $s_{j}\left[0\right]$.
\subsection{State space of RWIG}\label{section: G_t enumeration}
The fundamental assumption of RWIG is that any pair of walkers that meets at time $k$ in the same Markov state is connected in the contact graph $G_k$. In other words, if $q$ walkers reside in the same state in the Markov graph at discrete time $k$, they form a fully connected subgraph, i.e. clique of size $q$ in the contact graph $G_k$. Consequently, the graph $G_k$ consists of the union of disconnected cliques and $G_k$ is only connected and equal to a complete graph $K_M$ if all $M$ walkers meet in the same state. The induced structure describes the contact graph through pairwise disjoint subsets of walkers, which is exemplified in Figure \ref{fig:cliques as subsets}.

\begin{figure*}[!t]
\centering
\subfloat{\includegraphics[width=2.75in]{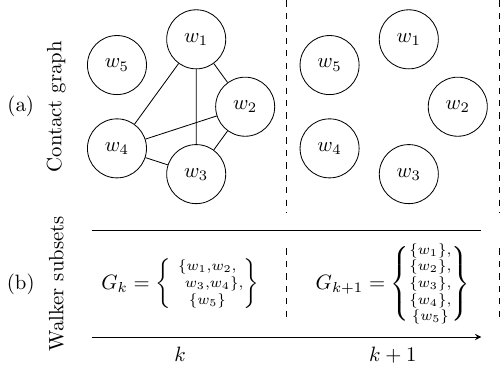}}%
\subfloat{\includegraphics[width=2.75in]{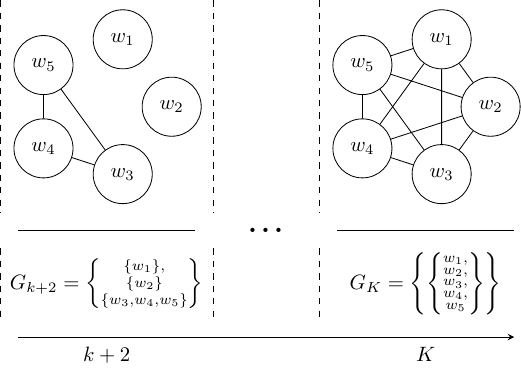}}%
\vspace{-0.15cm}\caption{Cliques (a) as partitions on the walker set (b).}
    \label{fig:cliques as subsets}
\end{figure*}

The union of the walker subsets in the node set of contact graph $G_k$ equals the complete walker set $\mathcal{M} = \{w_1, w_2, ..., w_M\}$. Since the subsets are pairwise disjoint, each possible contact graph generated by RWIG is equivalent to a partition on the walker set $\mathcal{M}$, whose number of cells is equal to the number of disconnected cliques. Thus, we refer to partitions on the walker set $\mathcal{M}$ and contact graphs interchangeably. Additionally, we also refer to $m$-partitions on the walker set $\mathcal{M}$ and $m$-clique contact graphs (i.e. a contact graph with $m$ cliques) interchangeably.

To count the number of possible contact graphs, consider an $m$-clique contact graph $G_k$ at some time $k$, which is equivalent to an $m$-partition $\pi_m$ on the walker set $\mathcal{M}$. In the contact graph, $M$ walkers occupy $m$ different states, where $m \leq M$. Additionally, the number of occupied states $m$ is upper bounded by the total number of states $N$ in the Markov graph. Therefore, the number of states occupied by walkers $m$ is upper bounded by $min(N,M)$. The total number of contact graphs $|G_k|$ is obtained by summing the number of all possible $m$-partitions
\begin{align}
    |G_k| = \sum_{m=0}^{min(N,M)}\mathcal{S}_M^{(m)},
\end{align}
where $\mathcal{S}_M^{(m)}$ are the \textit{Stirling numbers of the second kind}  \cite{book_abramowitz}.

If the number of walkers does not exceed the number of walker states (i.e. $M{\leq}N$), $m$ is upper bounded by $min(N,M){=}M$. Therefore, the total number of partitions on the walker set $\mathcal{M}$ and consequently, the total number of contact graphs is $|G_k|{=} \sum_{m=0}^M\mathcal{S}_M^{(m)}{=}\mathcal{B}_M$, where $\mathcal{B}_M$ is the \textit{$M$-th Bell number}. The Bell numbers are explained in Appendix \ref{apx:combinatorics}. Table \ref{tab: gk count for all mn} illustrates a few examples of the number of contact graphs for various combinations of walker count $M$ and number of Markov states $N$, where the regime $M{\leq}N$ is shaded. For instance, if \(M{=}5\) and \(N{=}3\), the total number of contact graphs is 41 as there are 25 ways for 5 walkers to occupy 3 states and form a 3-clique graph, 15 ways to occupy 2 states (\(G_k\) is a 2-clique graph) and only 1 way to be in the same state (\(G_k\) is a complete graph).

\begin{table}[!t]
    \centering
        \caption{Examples of contact graph state space cardinality \\with respect to $M$ walkers and $N$ states.}\vspace{-0.25cm}
    \label{tab: gk count for all mn}
    \begin{tabular}{|l||*{6}{c|}}\hline
    \backslashbox{$M$}{$N$}
    &\makebox[1em]{5}&\makebox[1em]{6}&\makebox[1em]{7}&\makebox[1em]{8}&\makebox[1em]{9}&\makebox[1em]{10}\\\hline\hline
    \makebox[3em]{1}&\cellcolor{gray!20}1&\cellcolor{gray!20}1&\cellcolor{gray!20}1&\cellcolor{gray!20}1&\cellcolor{gray!20}1&\cellcolor{gray!20}1\\\hline
    \makebox[3em]{2}&\cellcolor{gray!20}2&\cellcolor{gray!20}2&\cellcolor{gray!20}2&\cellcolor{gray!20}2&\cellcolor{gray!20}2&\cellcolor{gray!20}2\\\hline
    \makebox[3em]{3}&\cellcolor{gray!20}5&\cellcolor{gray!20}5&\cellcolor{gray!20}5&\cellcolor{gray!20}5&\cellcolor{gray!20}5&\cellcolor{gray!20}5\\\hline
    \makebox[3em]{4}&\cellcolor{gray!20}15&\cellcolor{gray!20}15&\cellcolor{gray!20}15&\cellcolor{gray!20}15&\cellcolor{gray!20}15&\cellcolor{gray!20}15\\\hline
    \makebox[3em]{5}&\cellcolor{gray!20}52&\cellcolor{gray!20}52&\cellcolor{gray!20}52&\cellcolor{gray!20}52&\cellcolor{gray!20}52&\cellcolor{gray!20}52\\\hline
    \makebox[3em]{6}&202&\cellcolor{gray!20}203&\cellcolor{gray!20}203&\cellcolor{gray!20}203&\cellcolor{gray!20}203&\cellcolor{gray!20}203\\\hline
    \makebox[3em]{7}&855&876&\cellcolor{gray!20}877&\cellcolor{gray!20}877&\cellcolor{gray!20}877&\cellcolor{gray!20}877\\\hline
    \makebox[3em]{8}&3845&4111&4139&\cellcolor{gray!20}4140&\cellcolor{gray!20}4140&\cellcolor{gray!20}4140\\\hline
    \makebox[3em]{9}&18002&20648&21110&21146&\cellcolor{gray!20}21147&\cellcolor{gray!20}21147\\\hline
    \makebox[3em]{10}&86472&109299&115179&115929&115974&\cellcolor{gray!20}115975\\\hline
    \end{tabular}

\end{table}

Therefore, if the number of walkers does not exceed the number of states $M{\leq}N$, the number of contact graphs formed by $M$ random walkers is equal to the Bell number $\mathcal{B}_M$. Otherwise, if $M{>}N$, then we omit $m$-partitions on the walker set $\mathcal{M}$ where $m{>}N$ because walkers cannot be found in more cliques than there are Markov states.

\subsection{Contact graph probability: Examples}
After the enumeration of contact graphs in Section \ref{section: G_t enumeration}, we now seek to find the probability distribution of the contact graphs $G_k$, conditioned on the initial state vector $s_{j}[0]$ of walker $w_j$ and Markov transition matrices $P_{j}$ for each walker $w_j \in \{w_1, w_2,\ldots, w_M\}$. 

A contact graph realisation with $m$ cliques is denoted as $g = \{\mathcal{A}_{1}, \mathcal{A}_{2},..., \mathcal{A}_{m}\}$, where $\mathcal{A}_{i}$ for all $i \in \{1, 2, ..., m\}$, represent the cliques formed at a discrete time step. Due to the equivalence between contact graphs and partitions on the walker set $\mathcal{M}$ shown in Figure \ref{fig:cliques as subsets}, the cliques $\mathcal{A}_{i}$ are functionally subsets of walkers found to be in the same state in the Markov graph at a given time.

We also introduce the set of initial conditions for all walkers: $\mathbf{s}_{\mathcal{M}}[0]{=}\{s_{j}[0]\}_{j=1}^M$, as well as the set of $N{\times}N$ transition probability matrices for all walkers $\textbf{P}_{\mathcal{M}}{=}\{P_{j}\}_{j=1}^M$.

\subsubsection{Introductory example}\label{section: introductory}
The simplest contact graph example is the complete graph $G_k{=}\{\mathcal{M}\}$, where all walkers are found in the same Markov state at discrete time $k$. 

The random variable $X_{j}[k]$ denotes the state in the Markov graph of walker $w_j$ at discrete time $k$ and $\Pr[X_{j}[k]{=}i]$ is the probability that walker $w_j$ is in state $i$ in the Markov graph at discrete time $k$. The $i$-th element of the probability state vector $s_{j}[k]$ for walker $w_j$ at time $k$ is then $(s_{j}[k])_i{=}\Pr[X_{j}[k]{=}i]$. Only if all $M$ walkers are in the same state at discrete time $k$, a complete graph $K_M$ is formed. The probability that all $M$ walkers are in state $i$ equals $\prod_{j=1}^M \Pr[X_{j}[k]{=}i]$, because all random walkers move independently of each other in the Markov graph. Summing the probabilities that all walkers are in state $i$ over all states $i \in \{1, 2, ..., N\}$ results in the probability that a complete graph $G_k{=}\{\mathcal{M}\}{\equiv}K_M$ is created at discrete time $k$:
\begin{align*}
    \Pr[G_k{=}\{\mathcal{M}\}] = \sum_{i=1}^N \prod_{j=1}^M \Pr[X_{j}[k] {=} i]
    =\sum_{i=1}^N \prod_{j=1}^M (s_{j}[k])e_i^T,
\end{align*}
where $e_i$ is the all-zero row vector with 1 at $i$-th position \cite{perfanalysis}. Introducing the Hadamard product \cite{PVM_graphspectra_second_edition} of the walkers' state probability vectors $s_{1}[k] \circ {...} \circ s_{M}[k] = \bigodot_{j=1}^M s_{j}[k]$:
\begin{align*}
    \Pr[G_k{=}\{\mathcal{M}\}] {=} \sum_{i=1}^N \prod_{j=1}^M (s_{j}[k])e_i^T{=}\sum_{i=1}^N  \left(\bigodot_{j=1}^M s_{j}[k]\right)e_i^T.
\end{align*}
Finally, introducing the all-ones vector $u = [1, \ldots, 1]$ yields:
\begin{align}\label{eq:sigma_m}
    \Pr[G_k{=}\{\mathcal{M}\}]  
    &=\left(\bigodot_{j=1}^M s_{j}[k]\right)u^T{=}\bigodot_{j=1}^M \left(s_{j}[0] P_{j}^k\right) u^T,
\end{align}
where we have used the $k$-step Markov transition probability formula \cite{perfanalysis}: $s_{j}[k] = s_{j}[0]P_{j}^k$.

Equation (\ref{eq:sigma_m}) expresses the probability that the $M$ walkers are in the same state in the Markov graph at discrete time $k$. Since $\mathcal{M}$ is just an example of any $M$ size walker set, equation (\ref{eq:sigma_m}) is also directly applicable to any walker subset $\mathcal{A}_{i} \subseteq \mathcal{M}$, where $i\in \{1, 2, ..., m\}$. We thus define $\sigma_{\mathcal{A}_{i}}[k]$ as the probability that walkers of a subset $\mathcal{A}_{i} \subseteq \mathcal{M}$ are in the same state at discrete time $k$:
\begin{equation}\label{eq:sigmaA}
      \sigma_{\mathcal{A}_{i}}[k] = \bigodot_{w_j \in \mathcal{A}_{i}} \left(s_{j}[0]P_{j}^k\right) u^T.
\end{equation}

The implementation of equation (\ref{eq:sigmaA}) is provided in Appendix \ref{apx: thm1 algo} (Algorithm \ref{alg:calc sigma}). The definition of $\sigma_{\mathcal{A}_{i}}[k]$ in (\ref{eq:sigmaA}) constitutes the basis of our further analysis, because equation (\ref{eq:sigmaA}) forms a compact and analytically tractable formula relating contact graph probabilities to the transition probability matrices and initial conditions. 

Equation (\ref{eq:sigma_m}) calculates the probability that all walkers are in the same state, which is equivalent to the probability of the $1$-clique contact graph or the complete graph $K_M$. To offer insight into the probability of contact graphs with more than one clique, we first calculate the probabilities of the $2$-clique and $3$-clique contact graphs. We then state and prove in Section \ref{section: main section} our main theorem for the probability of a general $m$-clique contact graph.

\subsubsection{2-clique contact graph}
    Let $g$ be a 2-clique contact graph realisation $g = \{\mathcal{A}_1, \mathcal{A}_2\}$. We consider that the walkers in cliques $\mathcal{A}_1$ and $\mathcal{A}_2$ are in Markov states $i$ and $j$ respectively. Summing over all states $i, j$ where $i \neq j$, the probability $\Pr[G_k {=} g]$ is:
\begin{align}\label{eq: 2 cliques}
    \Pr[G_k {=} g] &{=} \sum_{i=1}^N \sum_{\substack{j=1\\ j \neq i}}^N \left(\prod_{w_u \in \mathcal{A}_1} s_u[k]e_i^T  \right) \left(\prod_{w_v \in \mathcal{A}_2} s_v[k]e_j^T \right) \nonumber \\
    &{=} \sum_{i=1}^N \left(\prod_{w_u \in \mathcal{A}_1} s_u[k]e_i^T \right) \sum_{\substack{j=1\\ j \neq i}}^N \left(\prod_{w_v \in \mathcal{A}_2} s_v[k]e_j^T \right).
\end{align}
We rewrite the second sum-product term as:
\begin{align*}
    \sum_{\substack{j=1\\ j \neq i}}^N \prod_{w_v {\in} \mathcal{A}_2} s_v[k]e_j^T {=}\sum_{j=1}^N \left(\prod_{w_v {\in} \mathcal{A}_2} s_v[k]e_j^T  \right){-}\prod_{w_v {\in} \mathcal{A}_2}  s_v[k]e_i^T.
\end{align*}
Introducing the definition of $\sigma_\mathcal{A}[k]$ in (\ref{eq:sigmaA}) yields:
\begin{align}\label{substitution: sigmaA}
        \sum_{\substack{j=1\\ j \neq i}}^N \prod_{w_v \in \mathcal{A}_2} s_v[k]e_j^T &= \bigodot_{w_v \in \mathcal{A}_2} \left(s_{v}[0]P_v^k  \right)u^T - \prod_{w_v \in \mathcal{A}_2}  s_v[k]e_i^T\nonumber\\
        &= \sigma_{\mathcal{A}_2}[k] - \prod_{w_v \in \mathcal{A}_2} s_v[k]e_i^T.
\end{align}
Substituting (\ref{substitution: sigmaA}) into (\ref{eq: 2 cliques}):
\begin{align}
    \Pr[G_k {=} g]
    &=\sigma_{\mathcal{A}_2}[k] \sum_{i=1}^N \left(\prod_{w_u \in \mathcal{A}_1} s_u[k]e_i^T \right)\nonumber\\
    &\ \ \ - \sum_{i=1}^N \left(\prod_{w_u \in \mathcal{A}_1} s_u[k]e_i^T \prod_{w_v \in \mathcal{A}_2} s_v[k]e_i^T \right).\nonumber
\end{align}
Since $\mathcal{A}_1$ and $\mathcal{A}_2$ are complements w.r.t. the walker set $\mathcal{M}$, then $\mathcal{A}_1 \cup \mathcal{A}_2 = \mathcal{M}$ and thus:
\begin{align}
    \prod_{w_u \in \mathcal{A}_1} s_u[k]e_i^T \prod_{w_v \in \mathcal{A}_2} s_v[k]e_i^T = \prod_{w_u \in \mathcal{M}} s_u[k]e_i^T.\nonumber
\end{align}
Finally, by (\ref{eq:sigma_m}) and (\ref{eq:sigmaA}):
\begin{align}\label{eq:2cliques closedform 1}
    \Pr[G_k{=}g] &= \sigma_{\mathcal{A}_2}[k] \sum_{i=1}^N \prod_{w_u \in \mathcal{A}_1} s_u[k]e_i^T   {-} \sum_{i=1}^N \prod_{w_u \in \mathcal{M}} s_u[k]e_i^T \nonumber\\
    &=\sigma_{\mathcal{A}_1}[k]\sigma_{\mathcal{A}_2}[k] - \sigma_\mathcal{M}[k].
\end{align}

The intuition behind (\ref{eq:2cliques closedform 1}) is the inclusion-exclusion principle \cite[p. 10-12]{perfanalysis}, where the probability of 2 cliques is equal to the probability that walkers from the subsets $\mathcal{A}_1$ and $\mathcal{A}_2$ are each found in the same states minus (hence, excluding) the probability that all walkers are in the same state.
\subsubsection{3-clique contact graph}\label{section: 3-clique contact graph}
Let $g$ be a 3-clique contact graph realisation $g{=}\{\mathcal{A}_1, \mathcal{A}_2, \mathcal{A}_3\}$. We consider that the walkers in clique $\mathcal{A}_1$ are in Markov state $i_1$, the walkers in clique $\mathcal{A}_2$ are in Markov state $i_2$ and that the walkers in clique $\mathcal{A}_3$ are in Markov state $i_3$. Summing over all states $i_1, i_2, i_3$ where $i_1 \neq i_2 \neq i_3$, the probability of the realisation $g$ is:
\begin{align}\label{eq: 3 cliques}
    \Pr[G_k {=} g] &= \sum_{i_1=1}^N \sum_{\substack{i_2=1\\ i_2 \notin \{i_1\}}}^N \sum_{\substack{i_3=1\\ i_3 \notin \{i_1, i_2\}}}^N \prod_{j=1}^3 \prod_{w_{u} \in \mathcal{A}_j} s_{u}[k]e_{i_j}^T .
\end{align}

Expanding equation (\ref{eq: 3 cliques}) is possible by observing that the probability $\Pr[G_k {=} g]$ is equal to the product of clique probabilities minus the probability of \textit{any contact graph obtained by amassing cliques} (e.g. $\mathcal{A}_1 {\cup} \mathcal{A}_2$ or $\mathcal{A}_1 {\cup} \mathcal{A}_2 {\cup} \mathcal{A}_3$). In other words, the event that walkers from each of the cliques $\mathcal{A}_1$, $\mathcal{A}_2$, $\mathcal{A}_3$  are found in the same state, which has probability $\sigma_{\mathcal{A}_1}[k] \sigma_{\mathcal{A}_2}[k] \sigma_{\mathcal{A}_3}[k]$, encompasses the 3 events: walkers occupy the same state, walkers occupy two different states, walkers occupy three different states. The probability of a contact graph with three cliques $g {=} \{\mathcal{A}_1$, $\mathcal{A}_2$, $\mathcal{A}_3\}$ is
    \begin{align*}
        \Pr[G_k {=} g] = & \sigma_{\mathcal{A}_1}[k] \sigma_{\mathcal{A}_2}[k] \sigma_{\mathcal{A}_3}[k] - \Pr[G_k {=} \{\mathcal{M}\}] \nonumber \\
        &- \Pr[G_k {=} \{\mathcal{A}_1 \cup \mathcal{A}_2, \mathcal{A}_3\}] \nonumber \\
        &- \Pr[G_k {=} \{\mathcal{A}_1 \cup \mathcal{A}_3, \mathcal{A}_2\}] \nonumber \\
        &- \Pr[G_k {=} \{\mathcal{A}_2 \cup \mathcal{A}_3, \mathcal{A}_1\}].
    \end{align*}
Denoting amassed cliques as $\mathcal{A}_i \cup \mathcal{A}_j = \mathcal{A}_{ij}$:
\begin{align}\label{eq: 3cliques recursive}
    \Pr[G_k {=} g] = & \sigma_{\mathcal{A}_1}[k] \sigma_{\mathcal{A}_2}[k] \sigma_{\mathcal{A}_3}[k] - \Pr[G_k {=} \{\mathcal{M}\}] \nonumber \\
    &- \Pr[G_k {=} \{\mathcal{A}_{12}, \mathcal{A}_3\}] \nonumber \\
    &- \Pr[G_k {=} \{\mathcal{A}_{13}, \mathcal{A}_2\}] \nonumber \\
    &- \Pr[G_k {=} \{\mathcal{A}_{23}, \mathcal{A}_1\}].
\end{align}
The sigma notation in (\ref{eq:sigmaA}) extends to amassed cliques as
    $\sigma_{\mathcal{A}_{i_1, i_2, {...}, i_m}}[k] = \sigma_{\mathcal{A}_{i_1} \cup \mathcal{A}_{i_2} \cup \ldots \cup \mathcal{A}_{i_m}}[k]$ and yields:
\begin{align}\label{eq: 3cliques closedform}
    \Pr[G_k{=}g] &= \sigma_{\mathcal{A}_1}[k] \sigma_{\mathcal{A}_2}[k] \sigma_{\mathcal{A}_3}[k] -\sigma_{\mathcal{M}}[k]
    \nonumber \\
    &\ \ \ 
    -(\sigma_{\mathcal{A}_{12}}[k]\sigma_{\mathcal{A}_3}[k] {-} \sigma_{\mathcal{M}}[k]) 
    \nonumber \\
    &\ \ \ 
    -(\sigma_{\mathcal{A}_{13}}[k]\sigma_{\mathcal{A}_2}[k] {-} \sigma_{\mathcal{M}}[k]) 
    \nonumber \\
    &\ \ \ 
    -(\sigma_{\mathcal{A}_{23}}[k]\sigma_{\mathcal{A}_1}[k] {-} \sigma_{\mathcal{M}}[k]) 
    \nonumber \\
    &=\sigma_{\mathcal{A}_1}[k] \sigma_{\mathcal{A}_2}[k] \sigma_{\mathcal{A}_3}[k] {-} \sigma_{\mathcal{A}_{12}}[k]\sigma_{\mathcal{A}_3}[k]
    \nonumber \\
    &\ \ \ {-}  \sigma_{\mathcal{A}_{13}}[k]\sigma_{\mathcal{A}_2}[k] {-}  \sigma_{\mathcal{A}_{23}}[k]\sigma_{\mathcal{A}_1}[k]
    {+} 2\sigma_{\mathcal{M}}[k].
\end{align}

The probability of a 4-clique contact graph is provided in Appendix \ref{apx: 4 clique graph proba}.

\section{Contact graph probability distribution}\label{section: main section}
Let $g$ be any $m$-clique contact graph: $g = \{\mathcal{A}_1, \mathcal{A}_2, \ldots, \mathcal{A}_m\}$. Equations (\ref{eq: 2 cliques}) and (\ref{eq: 3 cliques}) can be extended to compute the probability of an $m$-clique contact graph. 

\begin{theorem}\label{thm: m cliques}
    The probability of an $m$-clique contact graph $g = \{\mathcal{A}_1, \mathcal{A}_2, ..., \mathcal{A}_m\}$ at discrete time $k$ is
    \begin{align}\label{eq: m cliques combinatorial}
    \Pr[G_k{=}g] = \sum_{i_1=1}^N \sum_{\substack{i_2=1\\i_2 \not\in \{i_1\}}}^N ... \sum_{\substack{i_m=1\\i_m \not\in \{i_l\}_{l=1}^{m-1}}}^N \prod_{j=1}^m \prod_{w_{u} \in \mathcal{A}_j} \left(s_{u}[k]\right)_{i_j}.
\end{align}
\end{theorem}

\begin{proof}
The probability that \(M\) walkers form a contact graph $g {=} \{\mathcal{A}_1, {\ldots}, \mathcal{A}_m\}$ at discrete time $k$ in \(m\) states \(i_1,{\ldots},i_m\), where \(|\{i_1,{\ldots},i_m\}|{=}m\), is equal to \(\prod_{j{=}1}^m \left(\prod_{w_{u} \in \mathcal{A}_j} \left(s_{u}[k]\right)_{i_j} \right)\). Summing over all different \(m\) states $\{i_j\}_{j{=}1}^m$ yields the probability of the realisation $g$.
\end{proof}

Theorem \ref{thm: m cliques} offers the probability of an $m$-clique contact graph from a combinatorial perspective. However, equation (\ref{eq: m cliques combinatorial}) requires to consider \(\frac{N!}{(N-m)!}\) combinations of states where \(M\) walkers may form cliques \(\mathcal{A}_1, \ldots, \mathcal{A}_m\), which lead to a combinatorial explosion for a large number of states \(N\). Therefore, we derive a closed form for $\Pr[G_k = g ]$.

\subsection{Amassed clique contact graphs}\label{section: main theorem}
We offer a formal definition of amassed clique graphs introduced in subsection \ref{section: 3-clique contact graph}, and subsequently illustrate how the process of amassing cliques allows us to  formulate our main theorem and expand equation (\ref{eq: m cliques combinatorial}).

Equation (\ref{eq: 3cliques recursive}) offers insight into the recursive nature of contact graphs probabilities and partitions: \textit{amassed clique graphs are a result of partitioning the contact graph $G_k$ and taking the union of walkers.} 
\begin{example}\label{example: partitioning a contact graph}
    By taking a 2-partition $\pi_2 = \{\mathcal{C}_1, \mathcal{C}_2\} = \{\{\mathcal{A}_1, \mathcal{A}_2\}, \{\mathcal{A}_3\}\}$ on the realisation $g {=} \{\mathcal{A}_1, \mathcal{A}_2, \mathcal{A}_3\}$ and taking the union of cliques $\mathcal{A}_1 \cup \mathcal{A}_2 {=} \mathcal{A}_{12}$, we obtain the amassed-clique contact graph $g(\pi_2) {=} \{\mathcal{A}_{12}, \mathcal{A}_3\}$.
Schematically, the generation of amassed clique contact graphs is shown in Figure \ref{fig: amassing process}.\vspace{-0.3cm}
\begin{figure}
    \centering
    \includegraphics[width=3.5in]{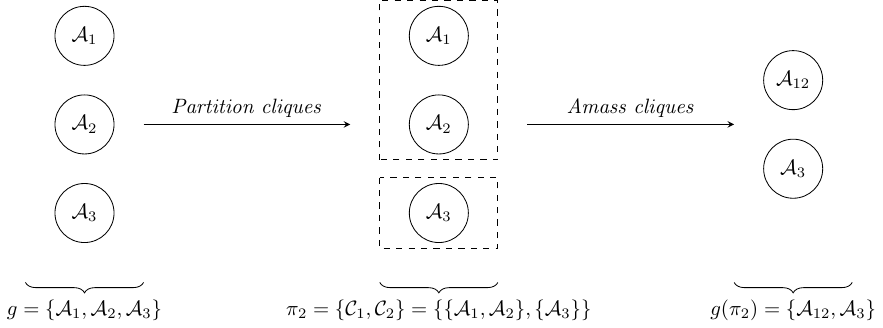}
    \vspace{-0.7cm}\caption{Process of creating amassed clique contact graphs.}
    \label{fig: amassing process}
\end{figure}
\end{example}
We call $g(\pi)$ the contact graph associated with partition $\pi$ on $g$. Naturally, the singleton partition $\pi_3 = \{\{\mathcal{A}_1\}, \{\mathcal{A}_2\}, \{\mathcal{A}_3\}\}$ has associated contact graph $g(\pi_3)=g=\{\mathcal{A}_1, \mathcal{A}_2, \mathcal{A}_3\}$. The rationale behind contact graphs generated by amassing cliques holds for any $m$-clique contact graph realisation $g {=} \{\mathcal{A}_1, \mathcal{A}_2, {...}, \mathcal{A}_m\}$, thus generalising equation (\ref{eq: 3cliques recursive}) to:
\begin{equation}\label{eq: kcliques recursive}
    \Pr[G_k {=} g ] = \prod_{k=1}^m \sigma_{\mathcal{A}_k} - \sum_{\pi \in \mathcal{P}^*_{g}} \Pr[G_k = g(\pi) ],
\end{equation}
where $\mathcal{P}_{g}$ is the set of all possible partitions on $g$ and $\mathcal{P}^{*}_{g}{=}\mathcal{P}_{g}{\setminus}\{\{\mathcal{A}_1\}, \{\mathcal{A}_2\}, {...}, \{\mathcal{A}_m\}\}$ excludes the singleton partition $\pi_m$.

We emphasize the distinction between partitions on a walker set and partitions on a contact graph. Recall from Section \ref{section: G_t enumeration} the equivalence relationship between contact graphs and partitions on the walkers set $\mathcal{M}$. Similarly, amassed clique graphs are a special class of contact graphs, which are obtained through partitioning cliques. We denote the difference by the symbol $\mathcal{C}$ for the cells of a partition on cliques and by $\mathcal{A}$ for the cells of a partition on the walker set $\mathcal{M}$. The caveat is illustrated in Figure \ref{fig: amassing process} of Example \ref{example: partitioning a contact graph}, where the cells of the 2-partition $\pi_2$ are: $\mathcal{C}_1 = \{\mathcal{A}_1, \mathcal{A}_2\}, \mathcal{C}_2 = \{\mathcal{A}_3\}$ and will be used in the proof of our main theorem. 

\subsection{Main theorem}
In equation (\ref{eq: kcliques recursive}), each $m$-clique graph realisation $g = \{\mathcal{A}_1, ..., \mathcal{A}_m\}$ probability depends on its \textit{associated sigma product} $\prod_{l=1}^m \sigma_{\mathcal{A}_l}[k]$, which allows us to reduce (\ref{eq: kcliques recursive}) to a closed form that depends only on sigma terms. Additionally, $\Pr[G_k {=} g ]$ also depends on the probability of graphs associated with all partitions on $g$. Thus, we are motivated to reduce equation (\ref{eq: kcliques recursive}) to a closed form:
\begin{align}\label{eq:sigma expansion}
    \Pr[G_k {=} g] = \sum_{\pi \in \mathcal{P}_{g}} \beta_m(\pi) \prod_{\mathcal{A} \in g(\pi)}\sigma_{\mathcal{A}}[k],
\end{align} 
where $\mathcal{A}$ is a clique in the amassed clique graph $g(\pi)$ (associated with partition $\pi$ on $g$), $\beta_m(\pi)\in \mathbb{Z}$ is the number of sigma product terms associated with contact graph $g(\pi)$ and subscript $m$ is the number of cliques in $g$. We call (\ref{eq:sigma expansion}) the \textit{sigma expansion} of (\ref{eq: kcliques recursive}) for contact graph $g$. We now state our main theorem: 
\begin{theorem}\label{thm: main theorem}
    The probability of an $m$-clique contact graph $g = \{\mathcal{A}_1, \mathcal{A}_2, ..., \mathcal{A}_m\}$ at discrete time $k$ is
    \begin{align}\label{eq: kcliques closed_form}
    \Pr[G_k {=} g] &= \sum_{\pi \in \mathcal{P}_{g}} \left(\prod_{\mathcal{C} \in \pi} ({-1})^{|\mathcal{C}|{-1}} (|\mathcal{C}|{-1})! \right)\prod_{\mathcal{A} \in g(\pi)}\sigma_{\mathcal{A}}[k],
\end{align}
where $|\mathcal{C}|$ denotes the number of cliques $\mathcal{A}$ in cell $\mathcal{C}$ of partition $\pi$ on $g = \{\mathcal{A}_1, \mathcal{A}_2, ..., \mathcal{A}_m\}$.
\end{theorem}

While Theorem \ref{thm: main theorem} may not be immediately intuitive, it offers a considerable advantage in terms of the runtime. We record the execution time of calculating the probability distribution of RWIG graphs using both equations (\ref{eq: m cliques combinatorial}) and (\ref{eq: kcliques closed_form}) and present the results in Appendix \ref{apx: time complexity}. The pseudocode for equation (\ref{eq: kcliques closed_form}) in Theorem \ref{thm: main theorem} is provided in Appendix \ref{apx: thm1 algo} (Algorithm \textit{RWIG-pmf}).
Our proof of Theorem \ref{thm: main theorem} stems directly from Lemmas \ref{lemma: trivial partition weight} and \ref{lemma: general}, presented below.
\begin{lemma}\label{lemma: trivial partition weight}
    Let $\pi_1 = \{\mathcal{M}\}$ be the $1$-partition on the walker set $\mathcal{M}$. The number $\beta_m(\pi_1)$  of sigma product terms $\sigma_\mathcal{M}[k]$ in the sigma expansion formula (\ref{eq:sigma expansion}) for the probability of an $m$-clique contact graph depends only on the number of cliques $m$ as
    \begin{equation}
        \beta_m(\pi_1) = (-1)^{(m-1)}(m-1)!
    \end{equation}
\end{lemma}

\begin{lemma}\label{lemma: general}
    Let $g$ be a  $m$-clique contact graph. Let $\pi_q = \{\mathcal{C}_1, {...}, \mathcal{C}_q\}$ be a $q$-partition on $g$, with $q < m$. Let the cardinality of each cell $\mathcal{C}_i$ be $c_i$. Let the number of sigma product terms $\prod_{i{=}1}^q \sigma_{\mathcal{C}_i}[k]$ in the sigma expansion formula of $g$ be $\beta_m(\pi_q)$. Then
    \begin{equation}
        \beta_m(\pi_q) = \prod_{i=1}^q (-1)^{c_i-1} (c_i-1)!
    \end{equation}
\end{lemma}

Lemma \ref{lemma: trivial partition weight} offers a formula for the weight $\beta_m(\pi_1)$ of the sigma product $\sigma_\mathcal{M}[k]$ (associated with the trivial partition $\pi_1 = \{\mathcal{M}\}$, i.e. the complete graph $K_M$). We build Lemma \ref{lemma: general} from Lemma \ref{lemma: trivial partition weight} as a generalisation from the trivial partition to any $q$-partition $\pi_q$ on the walker set $\mathcal{M}$. More precisely, we find the weight $\beta_m(\pi_q)$ of the sigma product associated with any $q$-partition on $g$, where $q < m$. The proofs of Lemmas \ref{lemma: trivial partition weight} and \ref{lemma: general} are provided in Appendix \ref{apx: lemma 1 2 proof}. 

The proof of Theorem \ref{thm: main theorem} is immediate by applying Lemma \ref{lemma: general} to all partitions on $g$.

\section{Steady-state contact graphs}\label{main section: steady-state graphs}
We assume that the same $N \times N$ Markov transition matrix $P$, which is common for all walkers, possesses a steady-state distribution $\Tilde{s}$, obeying $\Tilde{s} = \Tilde{s}P$. Then, the steady-state probability vector of each walker $w \in \mathcal{M}$ reduces to
\begin{align}\label{eq: steady-state common policy}
    \lim_{k \rightarrow \infty} s_w[k] = \Tilde{s}.
\end{align}

For a clique $\mathcal{A}$ of size $|\mathcal{A}|{=}q$ and recalling the $k$-step Markov transition probability $s_j[k]{=}s_j[0]P^k$, taking the limit in (\ref{eq:sigmaA}) as $k \rightarrow \infty$ and invoking the existence of a steady-state in (\ref{eq: steady-state common policy}) yields
\begin{align}\label{eq:sigma_a steadystate}
     \lim_{k \rightarrow \infty} \sigma_\mathcal{A}[k] &= \lim_{k \rightarrow \infty} \bigodot_{w_j \in \mathcal{A}} \left(s_{j}[0]P_{j}^k \right) u^T\nonumber\\
     &= \left(\bigodot_{w_j \in \mathcal{A}} \Tilde{s}\right) u^T= \sum_{i=1}^N (\Tilde{s}_i)^{q}.
\end{align}
The combinatorial nature of Theorem \ref{thm: main theorem} does not permit an analytical simplification of equation (\ref{eq: kcliques closed_form}) in the steady-state.
However, equation (\ref{eq:sigma_a steadystate}) illustrates that cliques of the same size have the same probability, because all walkers have the same steady-state distribution $\Tilde{s}$.  Therefore, the probability of a steady-state contact graph does not depend on the labelling of walkers inside cliques, but rather only on clique sizes and the steady-state vector $\Tilde{s}$.

\begin{example}
    Let $M{=}4$ walkers be in the steady-state $\Tilde{s}=[0.1\ 0.1\ 0.1\ 0.7]^T$. Using Theorem \ref{thm: main theorem}, we calculate the probability distribution of the steady-state contact graph $G_\infty$ formed by the walkers. In Figure \ref{fig:probadensity n=m=4}, we plot the most probable 4 realisations and illustrate that the second, third and fourth most probable realisations have equal probabilities and the same topologies.
\end{example}
    \begin{figure}[!t]
        \centering
        \includegraphics[width=\linewidth]{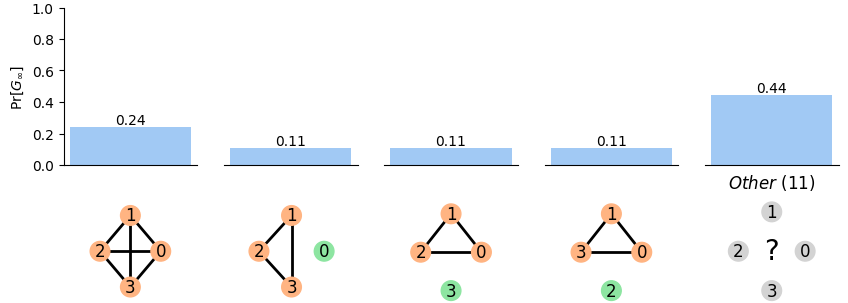}
        \vspace{-0.2cm}\caption{Most probable 4 realisations of the contact graph $G_\infty$ formed by 4 walkers.}
        \label{fig:probadensity n=m=4}
    \end{figure}
    \begin{figure}[!t]
        \centering
        \includegraphics[width=\linewidth]{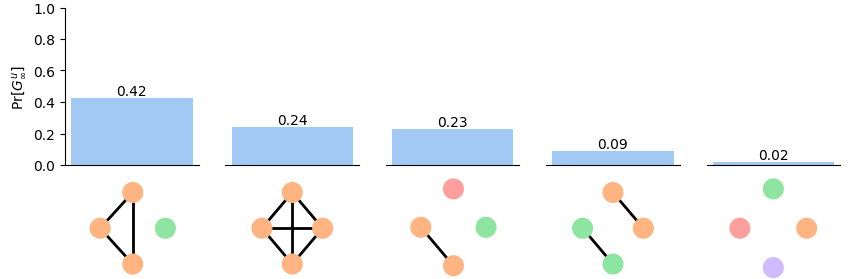}
        \vspace{-0.2cm}\caption{Probability density of the unlabelled contact graph $G^u_\infty$ formed by 4 walkers.}
        \label{fig:probadensity n=m=4 unlabelled}
    \end{figure}
Consequently, all $m$-clique steady-state contact graphs, which have the same $m$ clique sizes, have equal probability and we are thus motivated to study the probability of \textit{unlabelled} steady-state contact graphs.

\subsection{Unlabelled contact graphs}
Consider a steady-state $m$-clique contact graph realisation $g_\infty{=}\{\mathcal{A}_1,{\ldots}, \mathcal{A}_m\}$ with clique sizes $|\mathcal{A}_1|{=}q_1$, \ldots, $|\mathcal{A}_m|{=} q_m$. Additionally, denote by $\mathcal{Q} =\{q_1, {\ldots}, q_m\}$ a set of $m$ positive integers which sum to $M$, i.e.  $\sum_{i=1}^m q_i = M$.

We seek the number $\gamma(\mathcal{Q})$ of different steady-state $m$-clique contact graph realisations with $M$ walkers where the set formed by clique sizes of each realisation is equal to the set $\mathcal{Q}$.
The number $\gamma$ is equal to the solution to the combinatorial problem of counting \textit{how many ways there are to arrange $M$ identical objects into $m$ bins with sizes $\{q_1, {...}, q_m\}$}. If we denote by $c_j$ the number of cliques of size $j$ i.e. $c_j=|\{i\in \{1,...,m\}{:} q_i{=}j\}|$, for all $1 \leq j \leq M$, the solution to the problem is, by \cite[equation (13.3)]{book_combinatorics},
\begin{align}
    \gamma(\mathcal{Q}) = \frac{M!}{\prod_{i=1}^m q_i! \prod_{j=1}^M c_j!}\label{gamma_eq}.
\end{align}

Any realisation $g$ with equal clique size set has the same structure. Hence, removing the node labels of any contact graphs with  clique size set $\mathcal{Q}$ results in the \textit{unlabelled} graph $g_u$ which is equivalent to the set $\mathcal{Q}$. We consider unlabelled contact graphs when all walkers are found in the same steady-state because, as discussed in the beginning of Section \ref{main section: steady-state graphs}, all realisations with equal clique size set have equal probability. 

Unlabelled contact graphs allow us to scale RWIG to higher walker counts $M$ by reducing the contact graph state space. For instance, Table \ref{tab: gk count for all mn} shows that the total number of contact graphs for \(M{=}9\) walkers and \(N{=}10\) states is 21,147. However, the total number of \textit{unlabelled} contact graphs for \(M{=}9\) and \(N{=}10\) is only 30, which is the number of partitions of the positive integer \(M{=}9\) into a multiset of positive integers, such that the elements sum to \(M\). Hence, as we increase the number of walkers \(M\) considerably above the number of states $N$ in the Markov graph, we avoid the combinatorial Bell numbers explosion of the contact graph state space by omitting fine grained information on the walkers' clique assignment (i.e. which walker belongs to which clique) and allow for practical analysis of the clique sizes distribution.

Consider any labelled contact graph realisation $g$ which results in an unlabelled graph $g^u$. Then the probability of an unlabelled graph $g^u$ is defined by Lemma \ref{lemma: unlabelled contact graph proba}.

\begin{lemma}\label{lemma: unlabelled contact graph proba}
    The probability of a steady-state unlabelled $m$-clique contact graph $g^u$ with clique sizes $\mathcal{Q}{=}\{q_1, ..., q_m\}$ and $M$ walkers is

     \begin{equation}\label{eq: unlabelled graph proba from thm 1}
    \Pr[G^u_\infty = g^u] = \gamma(\mathcal{Q})\Pr[G_\infty {=} g],
\end{equation}
where $\gamma(\mathcal{Q})$ is defined by equation (\ref{gamma_eq}), $g {=} \{\mathcal{A}_1, {...}, \mathcal{A}_m\}$ is a contact graph realisation obtained by any labelling of the nodes in the unlabelled realisation $g^u$ with distinct labels from the walker set $\mathcal{M}$ and
\begin{equation}
\Pr[G_\infty{=}g]{=}\sum_{\pi \in \mathcal{P}_{g}} \left(\prod_{\mathcal{C} \in \pi} ({-1})^{{|\mathcal{C}|-1}} (|\mathcal{C}|{-1)}! \right)\prod_{\mathcal{A} \in g(\pi)} \left(\sum_{i=1}^N \Tilde{s}_i^{|\mathcal{A}|}\right)
\end{equation}
with all walkers traverse the same Markov graph with $N$ states and steady-state vector $\Tilde{s}$. 

\end{lemma}

\begin{example}
    Let $M {=} 4$ walkers be in the steady-state $\Tilde{s}{=}[0.1\ 0.1\ 0.1\ 0.7]^T$. Using Lemma \ref{lemma: unlabelled contact graph proba}, we calculate the probability distribution of the steady-state unlabelled contact graph $G^u_\infty$ formed by the walkers and plot it in Figure \ref{fig:probadensity n=m=4 unlabelled}.
\end{example}

\subsection{A combinatorial computation of the steady-state graph }
Another way to compute the probability of a steady-state unlabelled $m$-clique contact graph, where each walker has the same\footnote{If not all walker's probabilities are the same, then we must again compute all possible partitions as in Theorem \ref{thm: main theorem}.} steady-state vector \(\Tilde{s}\), can be obtained from Theorem \ref{thm: m cliques}. For a labelled contact graph $g {=} \{\mathcal{A}_1, {...}, \mathcal{A}_m\}$, the probability of realisation $g$ becomes
\begin{align}\label{eq: m cliques combinatorial steady-state}
    \Pr[G_\infty {=} g] &= \lim_{k \rightarrow \infty} \Pr[G_k = g ] =\nonumber\\
    &= \sum_{i_1=1}^N \sum_{\substack{i_2=1\\i_2 \not\in \{i_1\}}}^N ... \sum_{\substack{i_m=1\\i_m \not\in \{i_1, ..., i_{m-1}\}}}^N \prod_{j=1}^m \Tilde{s}_{i_j}^{q_j}.
\end{align}
where the clique sizes $q_i = |\mathcal{A}_i|$, for all $1 \leq j \leq m$, form the clique size set $\mathcal{Q} = \{q_1, {...}, q_m\}$. The number of labelled graphs with clique size set $\mathcal{Q}$ is given by equation (\ref{gamma_eq}), and thus the probability of an unlabelled steady-state contact graph $g^u$ with clique size set $\mathcal{Q}$ is 

\begin{align}\label{eq: m cliques combinatorial steady-state unlabelled}
    \Pr[G_\infty^u {=} g^u] = \frac{M! \sum_{i_1{=}1}^N \sum_{\substack{i_2{=}1\\i_2 \not\in \{i_1\}}}^N ... \sum_{\substack{i_m=1\\i_m \not\in \{i_l\}_{l=1}^{m{-}1}}}^N \prod_{j=1}^m \Tilde{s}_{i_j}^{q_j}}{\prod_{i=1}^m q_i! \prod_{j{=}1}^M c_j!}.
\end{align}
where $c_j$ is the number of cliques of size $j$, for all $1 {\leq} j {\leq} M$.

\section{Empirical analysis} \label{section: data}
The assumption of RWIG, that all walkers found in the same Markov state at discrete-time $k$ are connected in the contact graph, implies that the contact graph $G_k$ is formed by the union of disconnected cliques. In this section, we analyse various empirical temporal networks to validate our assumption and we demonstrate that RWIG is able to reproduce contact graphs with similar topological properties.
\subsection{Datasets}\label{section: datasets} 
We inspect a series of empirical datasets collected through the SocioPatterns sensing platform (http://www.sociopatterns.org). Génois and Barrat \cite{co-location} study how co-location graphs can be used as a proxy for face-to-face contacts. Similar to our fundamental assumption in RWIG, individuals are considered connected in a co-location graph if they are found to be in the same spatial location. Consequently, the co-location datasets are snapshots at discrete time steps of graphs. Our analysis of the datasets released in \cite{co-location} found that all co-location samples consist of unions of disconnected cliques, in complete accord with the topology of the contact graphs generated by RWIG.

We study the \textit{clique size} distribution and the \textit{clique count} distribution in co-location graphs. The clique size distribution is defined as the probability of observing a clique of a certain size and offers insight into possible patterns of \textit{typical} clique sizes. The clique count quantifies the connectivity of the contact graph. Figure \ref{fig: co-location maps} depicts the clique size and count distributions for three co-location datasets from \cite{co-location}: InVS15 (219 nodes), LyonSchool (242 nodes) and Thiers13 (328 nodes). For the clique size distribution, we only consider cliques formed of at least two individuals.

\begin{figure}[!t]
\centering
\includegraphics[width=2.54in]{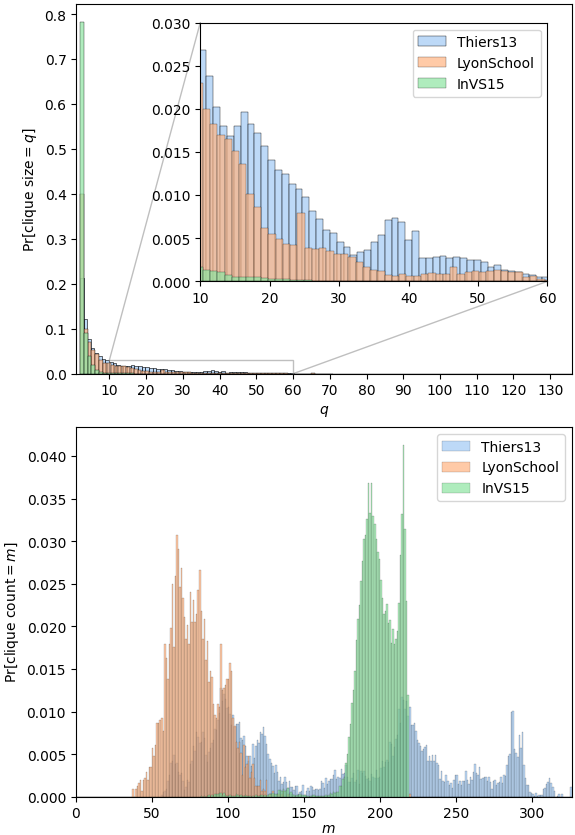}%
\centering
\vspace{-0.25cm}\caption{Clique size distribution (top) and clique count distribution (bottom) for three co-location datasets.}
    \label{fig: co-location maps}
\end{figure}

As shown in Figure \ref{fig: co-location maps}, most cliques for all datasets are small in size and consist mainly of two nodes. However, the clique count distribution indicates that some datasets exhibit stronger connectivity. For instance, the co-location graphs in LyonSchool have on average fewer cliques than the graphs in InVS15 while having a larger number of nodes. Hence, there is a comparably higher propensity for larger clique sizes in the graphs from LyonSchool, which is supported by Figure \ref{fig: co-location maps}. Overall, there is significant variability in the structure of co-location graphs.

\subsection{Simulations of steady-state contact graphs}\label{section: steady-state graphs}
We now show how RWIG is able to generate both sparse and dense contact graphs with minimal parameter tuning. As the unlabelled steady-state contact graph distribution $\Pr[G^u]$ depends only on the steady-state distribution $\Tilde{s}$, we compute the clique size and clique count distributions for a range of different steady-state distributions. We consider $M{=}10$ walkers and a Markov graph with $N{=}15$ states which admits a steady-state vector $\Tilde{s}$. We consider three different steady-state vectors (see Table \ref{tab:steady-states}). The first two steady-state vectors $\Tilde{s}{=}[s_1\ s_2\ ...\ s_N]^T$ have equal probability $s_1{=}s_2{=}...{=}s_{N{-}1}$ for the the first $N{-}1$ states while the probability $s_N$ of state $N$ takes values $s_N > s_1$ and $s_N \gg s_1$. We also consider the steady-state vector with the last three elements equal to each other $\Tilde{s} = [\frac{1}{1200}\ ...\ \frac{1}{1200}\ 0.32\ 0.32\ 0.32]$, which we call the \textit{Multimodal} steady-state vector. 

\begin{table}[h]
    \caption{Steady-state vectors.}\vspace{-0.3cm}
    \label{tab:steady-states}
    \centering
    \begin{tabular}{|l||*{1}{c|}}\hline
    \makebox[5em]{}&\makebox[2em]{$\Tilde{s}$}\\\hline\hline
    \makebox[5em]{$s_N = 0.33$}&[0.047 ... 0.047 0.33]\\\hline
    \makebox[5em]{$s_N = 0.96$}&[0.003 ... 0.003  0.96]\\\hline

    \makebox[5em]{$Multimodal$}&$[\frac{1}{1200}\ ...\ \frac{1}{1200}\ 0.32\ 0.32\ 0.32]$\\\hline
    \end{tabular}
\end{table}

 Figure \ref{fig: structure metrics steady-state} illustrates the clique size and clique count distributions for $M{=}10$ walkers on a $N{=}15$ state Markov graph. We overlay a smooth Kernel Density Estimate (KDE) line plot on top of the histograms for better visualisation.

\begin{figure}[!t]
\centering
\includegraphics[width=1.92in]{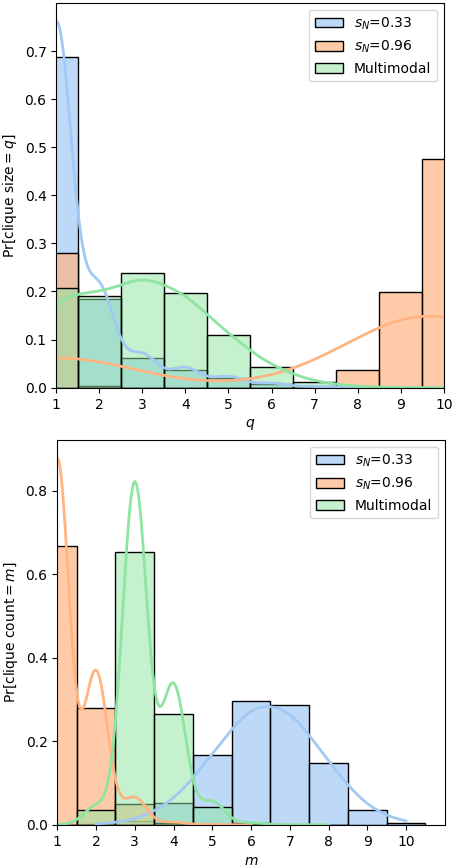}%
\centering
\vspace{-0.2cm}\caption{Clique size distribution (top) and clique count distribution (bottom) for unlabelled steady-state contact graphs.}
    \label{fig: structure metrics steady-state}
\end{figure}

Our experiments demonstrate that RWIG is capable of producing diverse contact graphs including graphs with many small cliques (e.g. $s_N {=} 0.33$) or few large cliques (e.g. $s_N {=} 0.96$), as shown in our empirical data analysis in Section \ref{section: datasets}. Furthermore, we have illustrated how clique size variety is already imposed by only tuning the steady-state vector $\Tilde{s}$. We expect that with generalisations such as different Markov graphs for each walker or more complex Markov graphs (e.g. reducible, periodic, several absorbing states, etc.), RWIG is capable of generating more tailored contact graph distributions and sequences.

\section{Conclusion}
We presented RWIG, a temporal contact graph model generated by independent random walkers on a Markov graph. A random walk on a Markov graph is a realization of a Markov process, which is specified in discrete time by a transition probability matrix $P_w$ and an initial condition $s_w[0]$ for each walker $w$. Hence, by choosing the matrices $P_w$ as well the vectors $s_w[0]$, any collection of discrete-time Markov processes can generate a corresponding temporal contact graph sequence consisting of disjoint cliques, which makes RWIG very general.  We derived the probability distribution (Theorem \ref{thm: main theorem}) of the RWIG contact graphs under the assumption of known initial walker states and transition probabilities in the Markov graph. We further analysed in Section \ref{section: data} through simulations the influence of a common walker steady-state vector on the distribution of generated contact graphs and illustrated that our model can produce both sparse and dense contact graphs. The analytical tractability of the model, along with the capability to create a wide variety of contact graphs, renders RWIG a promising basis for temporal graphs generative modelling.

\section{Further work}

We will explore extensions to the RWIG model:

First, RWIG generates links in the contact graph between walkers in the same state and necessarily generates graphs, which are unions of disconnected cliques.
We will elevate this limitation by extending RWIG to generate more complex subgraphs than cliques. 

Second, we plan to address the inverse problem, that consists of finding the class of transition probability matrices $P$ that generates a given $K$-length sequence of contact graphs $G_1, ..., G_K$. 
While statistical methods such as maximum likelihood estimation lie at the heart of the problem, the complexity of the parameter search space and scarcity of similarity measures for temporal graphs make this task non-trivial.

Third, given that a link in $G_k$ occurs, what is the probability that that link still exists at time $l{>}k$ in $G_l$? Alternatively, can RWIG's transition probability matrix be tuned to generate a "link burst" (i.e. the existence of a link over multiple time slots).
Many other questions or assumptions made in the temporal graph community may be addressed from the "process point of view" of RWIG.

Finally, motivated by the importance of higher-order correlations in time-respecting paths \cite{Scholtes2017,Lambiotte2019}, can RWIG be used to analytically calculate the probability of time-respecting paths of length $k$?
The answer would not only unravel which mechanisms (in terms of the underlying Markov graph and the transition probability matrix $P$) can lead to temporal graphs, whose \emph{causal topology} --i.e. which nodes can indirectly influence each other via time-respecting paths-- differs from that of the corresponding static graph, but
it would also shed light on the question why many human contact patterns exhibit second-order correlations, which has been shown to strongly influence the dynamics of diffusion and epidemic spreading \cite{Pfitzner2013,Scholtes2014}.

\section*{Acknowledgement}
This research has been funded by the European Research Council (ERC) under the European Union’s Horizon 2020 research and innovation program (grant agreement No 101019718). Ingo Scholtes acknowledges funding by the Swiss National Science Foundation (Grant No. 176938).

\ifCLASSOPTIONcaptionsoff
   \newpage
 \fi

\vskip -1\baselineskip plus -1fil
\begin{IEEEbiography}[{\includegraphics[width=1in,height=1.25in,clip,keepaspectratio]{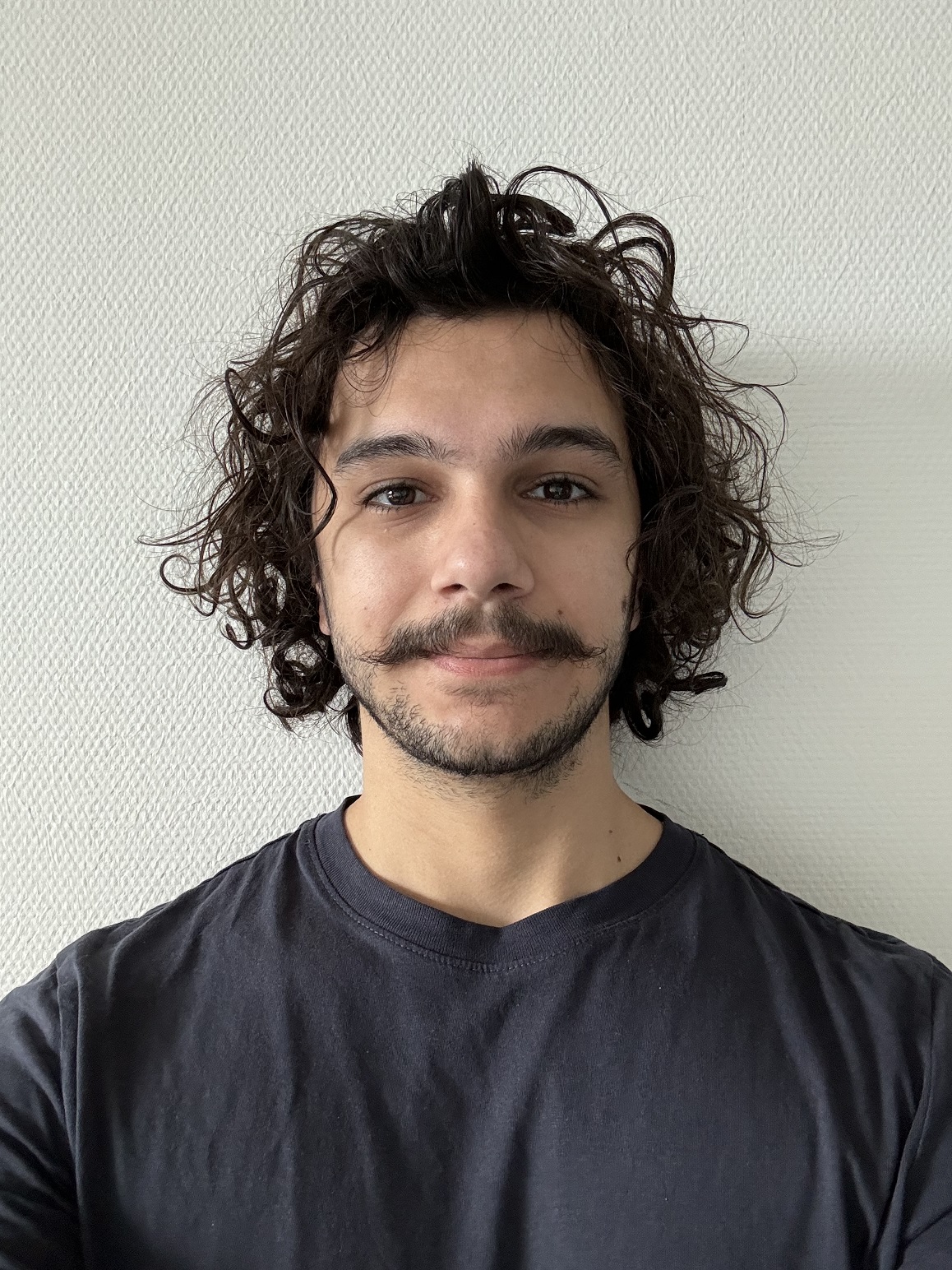}}]{Anton-David Almasan}
is a PhD student at the Delft University of Technology since 2024. He received a Master of Engineering degree from the University of Cambridge, with specialisation in Information and Computer Engineering. His main research interests span machine learning, network science and stochastic processes. Before joining Delft, David worked as an AI Consultant at Deepsea Technologies and as a Data Scientist at Thales UK.
\end{IEEEbiography}

\vskip -1\baselineskip plus -1fil
\begin{IEEEbiography}[{\includegraphics[width=1in,height=1.25in,clip,keepaspectratio]{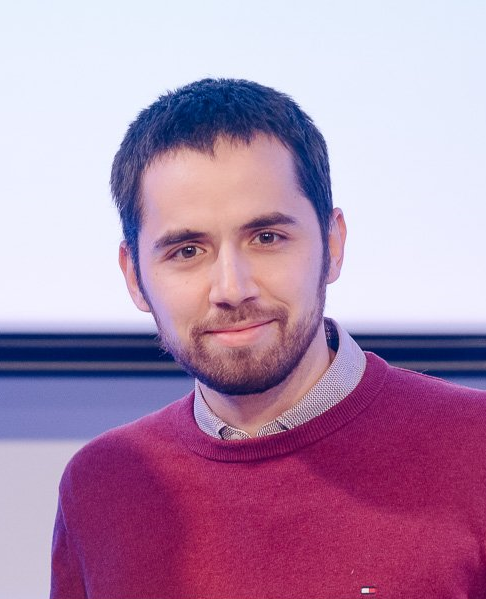}}]{Sergey Shvydun} received the Masters degree in Business Informatics (with distinction, 2014) and PhD in Applied Mathematics (cum laude, 2020) from the HSE University, Moscow, Russia. He is postdoctoral researcher at the Delft University of Technology since 2023. His main research interests include network science, machine learning, operations research, and social choice theory. Before joining Delft, he worked as an associate professor at the HSE University and as a senior research fellow at the Institute of
Control Sciences of the Russian Academy of Science. Sergey is the author of more than 30 papers in peer-reviewed journals and edited volumes and 1 book on centrality in networks.
\end{IEEEbiography}

\vskip -1\baselineskip plus -1fil
\begin{IEEEbiography}[{\includegraphics[width=1in,height=1.25in,clip,keepaspectratio]{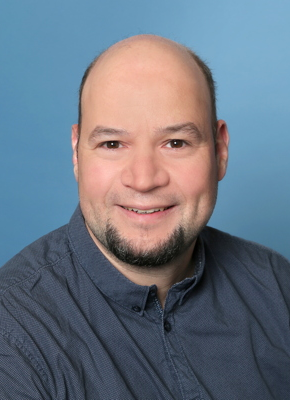}}]{Ingo Scholtes} is professor of Machine Learning for Complex Networks at the Center for Artificial Intelligence and Data Science (CAIDAS) at University of Würzburg and SNF Professor for Data Analytics at the Department of Informatics (IfI) at University of Zurich. 
His main research interests include network science, the modelling and analysis of temporal networks, and applications of machine learning to complex networks.
Apart from methodological contributions at the interface between network science and machine learning, his current research interests also include applications of (temporal) graph learning in collaborative software engineering, single cell biology, and astrophysics.
He is vice-spokesperson of the Center for Artificial Intelligence and Data Science (CAIDAS) and founding co-chair of the topical section Computational Social Science of the German Informatics Society (GI e.V.). 

Professor Scholtes received his Diploma (with distinction, 2005) and Doctorate degree in Computer Science (summa cum laude, 2011) from the University of Trier, Germany. 
From 2011 until 2018 he was a postdoctoral researcher and lecturer at the Chair of Systems Design at ETH Zürich.
In 2014, he was awarded a Junior-Fellowship from the German Informatics Society (GI e.V.).
In 2018 he was awarded an SNSF-Professorship from the Swiss National Science Foundation.
Before joining University of Würzburg as a chaired professor in 2021, he was professor for Data Analytics at the University of Wuppertal, Germany (2019 -- 2021) and at the University of Zurich (2018 - 2024).

Professor Scholtes is author of more than 80 papers in peer-reviewed conference proceedings and interdisciplinary journals.
He has served on the editorial board of \emph{Advances in Complex Systems} and is currently an associate editor of \emph{EPJ Data Science}.
\end{IEEEbiography}

\vskip -1\baselineskip plus -1fil
\begin{IEEEbiography}[{\includegraphics[width=1in,height=1.25in,clip,keepaspectratio]{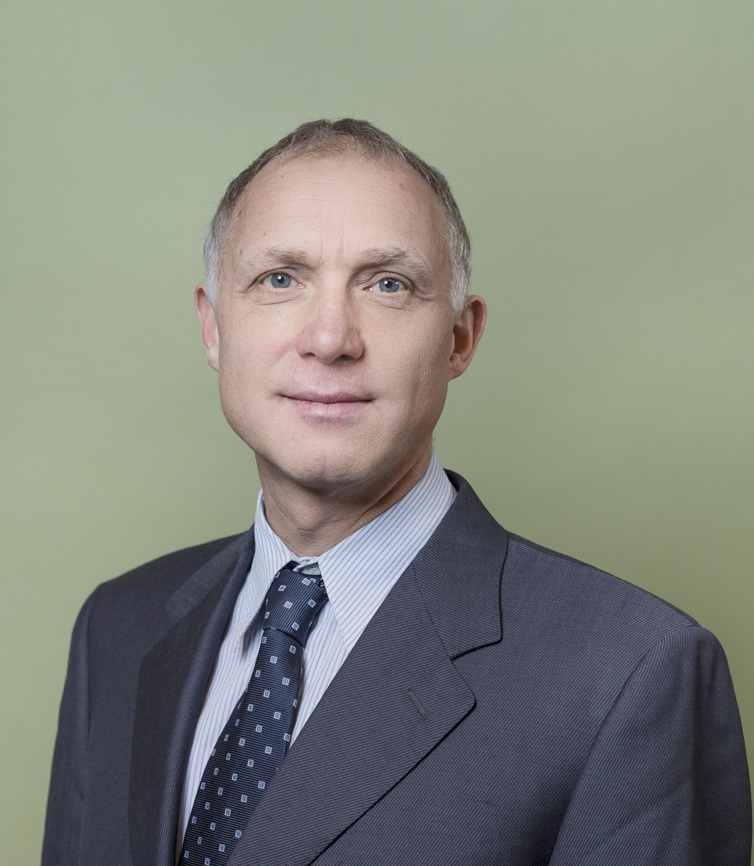}}]{Piet Van Mieghem} is professor at the Delft University of Technology and chairman of the section Network Architectures and Services (NAS) since 1998. He is the author of four books: Performance Analysis of Communications Networks and Systems, Data Communications Networking, Graph Spectra for Complex Networks and Performance Analysis of Complex Networks and Systems. He is a board member of the Netherlands Platform of Complex Systems, a steering committee member of the
Dutch Network Science Society, an external faculty member at the Institute for Advanced Study (IAS) of the University of Amsterdam and an IEEE Fellow. He was awarded an Advanced ERC grant 2020 for ViSiON, Virus Spread in Networks. Professor Van Mieghem received a Master degree and a Ph.D. degree in Electrical Engineering from the K.U.Leuven (Belgium) in 1987 and 1991, respectively. Before joining Delft, he worked at the Interuniversity Micro Electronic Center (IMEC) from 1987 to 1991. During 1993 to 1998, he was a member of the Alcatel Corporate Research Center in Antwerp. He was a visiting scientist at MIT (1992-1993) and a visiting professor at UCLA (2005), at Cornell University (2009), at Stanford University (2015) and at Princeton University (2022). Currently, he serves on the editorial board of the OUP Journal of Complex Networks. He was member of the editorial board of Computer Networks (2005-2006), the IEEE/ACM Transactions on Networking (2008- 2012), the Journal of Discrete Mathematics (2012- 2014) and Computer Communications (2012-2015).
\end{IEEEbiography}

\clearpage
\appendices
\section{List of symbols}\label{apx: symbols}
\begin{tabular}{p{0.34\linewidth}p{0.55\linewidth}}
    $N \in \mathbb{N}$&Number of states\\
    $K \in \mathbb{N}$&Discrete time\\
    $M \in \mathbb{N}$&Number of random walkers\\
    $P_w \in \mathbb{R}^{N \times N}$&Transition probability matrix of walker $w$\\
    $e_i \in \{0, 1\}^{1 \times N}$&The all zero vector with one at position $i$\\
    $u \in \{1\}^{1 \times N}$&The all ones vector\\ 
    $w_i$&Label of walker $i$\\
    $x_{j}[k] \in \{1,{\ldots}, N\}$ &State of walker $w_j$ at discrete time $k$\\
    $X_{j}[k]$&Random variable of the state of walker $w_j$ at discrete time $k$\\
    $s_{i}[k] \in \mathbb{R}^{1 \times N}$&Markov states probability distribution of walker $w_i$ at discrete time $k$\\
    $\mathcal{M} \in \{w_1,{\ldots}, w_M\}$&The complete set of $M$ walkers\\
    $|\mathcal{M}| = M$&The cardinality of set $\mathcal{M}$\\
    $\mathcal{A}_{i} \in \mathcal{M}$&Clique $i$ of walkers (i.e. a subset of the walker set $\mathcal{M}$)\\
    $\mathbf{s}_{\mathcal{M}}[k] \in \mathbb{R}^{M \times N}$&Markov states probability distribution of all walkers $\mathcal{M}$ at discrete time $k$\\   
    $\mathbf{s}_{\mathcal{A}_i}[k] \in \mathbb{R}^{|\mathcal{A}_{i}| \times N}$&Markov states probability distribution of walkers belonging to clique $\mathcal{A}_{i}$\\  
   $ \textbf{P}_\mathcal{A} \in \mathbb{R}^{N \times N \times |\mathcal{A}|}$&Tensor with transition probabilities of all walkers from subset $\mathcal{A}$ found in the same Markov state at time $k$\\
    $\pi_m$&$m$-partition on the walker set $\mathcal{M}$\\ 
    $g$&Contact graph (i.e. partition on the set $\mathcal{M}$) realisation\\ 
    $G_k$&Random variable of the contact graph at discrete time $k$\\ 
    $K_M$&Complete graph with $M$ nodes\\
    $\mathcal{S}_M^{(k)} \in \mathbb{N}$ &Stirling number of the second kind\\ 
    $\mathcal{B}_M \in \mathbb{N}$ &$M^{th}$ Bell number\\ 
    $\sigma_{\mathcal{A}_{i}}[k]$ &The probability that all walkers in the subset $\mathcal{A}_{i}$ are in the the same Markov state at time $k$ (i.e. at the same node in the underlying graph)\\ 
    $\mathcal{P}_g$ &The set of all possible partitions on the contact graph $g$\\ 
    $\mathcal{C}$ &Cell of a partition on a contact graph (i.e. a subset of cliques or a subset of subsets of walkers)\\
    $g(\pi)$ &The contact graph associated with partition $\pi$ on the contact graph $g$ (i.e. the amassed-clique contact graph formed by coalescing cliques in $g$)\\ 
    $\gamma(\mathcal{Q})$ &The number of labelled contact graphs with clique size set $\mathcal{Q}$\\
\end{tabular}
\begin{tabular}{p{0.35\linewidth}p{0.55\linewidth}}
    $\beta_m(\pi)$ &The weight of the sigma product associated with partition $\pi$ on a contact graph $g$ in the sigma expansion formula for the probability of contact graph $g$\\ 
    
\end{tabular}
\section{Prerequisite on combinatorics}\label{apx:combinatorics}
\subsection{Partitions and Stirling numbers}\label{apx: bell,stirling definitions}
\textit{Partition} \cite{PVM_graphspectra_second_edition_app,book_combinatorics_app}. A splitting of the elements of the set $\mathcal{M} = \{1, 2, ..., M\}$ into $m$ non-empty disjoint subsets is an $m$-partition $\pi_m$. The subsets of $\pi_m$ are called \textit{cells}. We refer to a partition $\pi$ where we do not know the number of cells by omitting the subscript $m$.
\begin{example}
    Let $\mathcal{M} = \{1, 2, 3, 4, 5\}$. A $3$-partition $\pi_3$ on $\mathcal{M}$ is $\{\{1, 2\}, \{3\}, \{4, 5\}\}$. The subsets $\mathcal{C}_1 = \{1, 2\}, \mathcal{C}_2 = \{3\}, \mathcal{C}_3 = \{4, 5\}$ are the 3 cells of $\pi_3$.
\end{example}
\textit{Stirling number of the second kind} \cite{book_combinatorics_app,book_comtet}. The Stirling number of the second kind $\mathcal{S}_M^{(k)}$ counts the number of $k$-partitions possible on a set with $M$ elements. Hence, $\mathcal{S}_M^{(k)} = 0 $ for $k > M$ and $\mathcal{S}_M^{(0)} = \delta_{0,M}$, where $\delta$ is the Kronecker delta. A few examples of Stirling numbers are presented in Table \ref{tab:stirling numbers example}.
\begin{table}[H]
    \centering
    \caption{Examples of $\mathcal{S}_N^{(k)}$.}
    \begin{tabular}{|l||*{6}{c|}}\hline
    \backslashbox{N}{k}&
    \makebox[2em]{5}&\makebox[2em]{6}&\makebox[2em]{7}&\makebox[2em]{8}&\makebox[2em]{9}&\makebox[2em]{10}\\\hline\hline
    \makebox[3em]{1}&1&&&&&\\\hline
    \makebox[3em]{2}&1&1&&&&\\\hline
    \makebox[3em]{3}&1&3&1&&&\\\hline
    \makebox[3em]{4}&1&7&6&1&&\\\hline
    \makebox[3em]{5}&1&15&25&10&1&\\\hline
    \makebox[3em]{6}&1&31&90&65&15&1\\\hline
    \makebox[3em]{7}&1&63&301&350&140&21\\\hline
    \makebox[3em]{8}&1&127&966&1701&1050&266\\\hline
    \makebox[3em]{9}&1&255&3025&7770&6951&2646\\\hline
    \makebox[3em]{10}&1&511&9330&34105&42525&22827\\\hline
    \end{tabular}
    \label{tab:stirling numbers example}
\end{table}

\subsection{Bell numbers}\label{apx: bell numbers}
\textit{Bell number} \cite{book_combinatorics_app,book_comtet}. The number of all possible partitions $\pi$ on a set with $M$ elements is the Bell number $\mathcal{B}_M$, which equals the sum of the number of $k$-partitions, for all $k = \{0, 1, ..., M\}.$ In terms of Stirling numbers of the second kind, the Bell Number $\mathcal{B}_M$ is:
\begin{align*}
    \mathcal{B}_M = \sum_{k=0}^M \mathcal{S}_M^{(k)}.
\end{align*}

The first few Bell numbers are presented in Table \ref{tab:bell numbers example}.
\begin{table}[H]
    \centering
    \caption{Bell numbers $\mathcal{B}_M$.}
    \begin{tabular}{|l||*{7}{c|}}\hline
    \makebox[3em]{$M$}&\makebox[2em]{1}&\makebox[2em]{2}&\makebox[2em]{3}&\makebox[2em]{4}&\makebox[2em]{5}&\makebox[2em]{...}&\makebox[2em]{10}\\\hline\hline
    \makebox[3em]{$\mathcal{B}_M$}&1&2&5&15&52&...&115975\\\hline
    \end{tabular}
    
    \label{tab:bell numbers example}
\end{table}

The $k$-th Bell number $\mathcal{B}_{k}$ equals the number of graphs on
$k$-nodes whose subgraphs consists of disconnected cliques. The Bell numbers satisfy the recursion%
\begin{equation}
\mathcal{B}_{n+1}=\sum_{k=0}^{n}\binom{n}{k}\mathcal{B}_{k},
\label{recursion_Bell_numbers}%
\end{equation}
where $\mathcal{B}_{0}=\mathcal{B}_{1}=1$. Since the binomial coefficients
$\binom{n}{k}$ are integers, the recursion (\ref{recursion_Bell_numbers})
indicates that the Bell numbers are also integers. For example, apart
$\mathcal{B}_{0}=\mathcal{B}_{1}=1$, we find that $\mathcal{B}_{2}=2$,
$\mathcal{B}_{3}=5$, $\mathcal{B}_{4}=15$, $\mathcal{B}_{5}=52$,
$\mathcal{B}_{6}=203$, $\mathcal{B}_{7}=877$, $\mathcal{B}_{8}=4140$,
$\mathcal{B}_{9}=21147$ and $\mathcal{B}_{10}=115975$. Another form of the
recursion (\ref{recursion_Bell_numbers}) is%
\[
\frac{\mathcal{B}_{n+1}}{n!}=\sum_{k=0}^{n}\frac{\mathcal{B}_{k}}{k!}\frac
{1}{\left(  n-k\right)  !},%
\]
which motivates us to consider the generating function $F\left(  z\right)
=\sum_{k=0}^{\infty}\frac{\mathcal{B}_{k}}{k!}z^{k}$ of the Bell numbers. We
multiply both sides of the rewritten recusion by $z^{n}$ and summing over all
non-negative integer $n\geq0$,
\[
\sum_{n=0}^{\infty}\frac{\mathcal{B}_{n+1}}{n!}z^{n}=\sum_{n=0}^{\infty
}\left(  \sum_{k=0}^{n}\frac{\mathcal{B}_{k}}{k!}\frac{1}{\left(  n-k\right)
!}\right)  z^{n}.%
\]
The Cauchy product of two power series (see e.g. \cite{PVM_charcoef})%
\[
\sum_{n=0}^{\infty}a_{n}z^{n}\sum_{n=0}^{\infty}b_{n}z^{n}=\sum_{n=0}^{\infty
}\left(  \sum_{k=0}^{n}a_{k}b_{n-k}\right)  z^{n}%
\]
indicates that%
\[
\sum_{n=0}^{\infty}\frac{\mathcal{B}_{n+1}}{n!}z^{n}=\sum_{k=0}^{\infty}%
\frac{\mathcal{B}_{k}}{k!}z^{k}\sum_{k=0}^{\infty}\frac{1}{k!}z^{k}=F\left(
z\right)  e^{z}.%
\]
Differentiating the generating function yields $\frac{dF\left(  z\right)
}{dz}=\sum_{k=1}^{\infty}\frac{\mathcal{B}_{k}}{\left(  k-1\right)  !}%
z^{k-1}=\sum_{n=0}^{\infty}\frac{\mathcal{B}_{n+1}}{n!}z^{n}$ which shows that
the generating function obeys%
\[
\frac{dF\left(  z\right)  }{dz}=F\left(  z\right)  e^{z}%
\]
or $\frac{1}{F\left(  z\right)  }\frac{dF\left(  z\right)  }{dz}=\frac{d}%
{dz}\log F\left(  z\right)  =e^{z}$. After integration with respect to $z$
from 0 to $z$, we find%
\[
\log F\left(  z\right)  -\log F\left(  0\right)  =e^{z}-1
\]
which equals
\[
F\left(  z\right)  =F\left(  0\right)  e^{\left(  e^{z}-1\right)  }.%
\]
Since $F\left(  0\right)  =\mathcal{B}_{0}=1$, we finally  arrive at the explicit form of the generating function of the Bell
numbers%
\begin{equation}
F\left(
z\right)=\sum_{k=0}^{\infty}\frac{\mathcal{B}_{k}}{k!}z^{k}=e^{\left(  e^{z}-1\right)
}=\frac{1}{e}e^{e^{z}}. \label{gf_Bell_numbers}%
\end{equation}

The Taylor expansion of $e^{e^{z}}$ around $z_{0}=0$ is $e^{e^{z}}=\sum
_{n=0}^{\infty}\frac{e^{nz}}{n!}=\sum_{n=0}^{\infty}\frac{1}{n!}\sum
_{k=0}^{\infty}\frac{n^{k}}{k!}z^{k}$ and%
\[
e^{e^{z}}=\sum_{k=0}^{\infty}\left(  \sum_{n=0}^{\infty}\frac{n^{k}}%
{n!}\right)  \frac{z^{k}}{k!}.%
\]
Equating corresponding powers in the above and (\ref{gf_Bell_numbers}) gives
the infinite series for the Bell numbers%
\begin{equation}
\mathcal{B}_{k}=\frac{1}{e}\sum_{n=0}^{\infty}\frac{n^{k}}{n!}
\label{Series_Bell_number}.%
\end{equation}

\subsection{Bell numbers and Stirling numbers}

The \textquotedblleft double\textquotedblright\ generating functions of the
Stirling numbers of the first kind $S_{m}^{(k)}$ and of the second kind $\mathcal{S}%
_{m}^{(k)}$ are \cite{PVM_charcoef}%
\begin{align}
\left(  1+x\right)  ^{u} &  =\sum_{m=0}^{\infty}\sum_{k=0}^{m}\frac
{S_{m}^{(k)}}{m!}u^{k}x^{m},\label{double_genfunc_Stirling_first_kind}\\
e^{u\left(  e^{x}-1\right)  } &  =\sum_{m=0}^{\infty}\sum_{k=0}^{m}%
\frac{\mathcal{S}_{m}^{(k)}}{m!}u^{k}x^{m}.\nonumber
\end{align}
Comparing with (\ref{gf_Bell_numbers}) shows that%
\[
\sum_{m=0}^{\infty}\frac{\mathcal{B}_{m}}{m!}z^{m}=e^{\left(  e^{z}-1\right)
}=\sum_{m=0}^{\infty}\left(  \sum_{k=0}^{m}\frac{\mathcal{S}_{m}^{(k)}}%
{m!}\right)  z^{m}.%
\]
Equating corresponding powers in $z$ results in%
\begin{equation}
\mathcal{B}_{m}=\sum_{k=0}^{m}\mathcal{S}_{m}^{(k)}.%
\label{Bell_number_as_sum_StirlingS2}%
\end{equation}
Finally, invoking the closed form of the Stirling numbers of the second kind
$\mathcal{S}_{m}^{(k)}$ (see e.g. \cite[sec. 24.1.4.C]{book_abramowitz_app})
\begin{equation}
\mathcal{S}_{m}^{(k)}=\frac{1}{k!}\sum_{j=0}^{k}(-1)^{k-j}{\binom{k}{j}}%
j^{m}\label{stirling2closed}%
\end{equation}
then leads to the finite sum%
\[
\mathcal{B}_{m}=\sum_{k=0}^{m}\frac{1}{k!}\sum_{j=0}^{k}(-1)^{k-j}{\binom
{k}{j}}j^{m}.%
\]
Simplifying $\mathcal{B}_{m}=\sum_{k=0}^{m}\sum_{j=0}^{k}(-1)^{k-j}\frac
{j^{m}}{j!\left(  k-j\right)  !}=\sum_{j=0}^{m}(-1)^{j}\frac{j^{m}}{j!}%
\sum_{k=j}^{m}\frac{\left(  -1\right)  ^{k}}{\left(  k-j\right)  !}$ becomes
\begin{equation}
\mathcal{B}_{m}=\sum_{j=0}^{m}\frac{j^{m}}{j!}\sum_{k=0}^{m-j}\frac{\left(
-1\right)  ^{k}}{k!},\label{Bell_closed}%
\end{equation}
where we observe that $\sum_{k=0}^{m-j}\frac{\left(  -1\right)  ^{k}}{k!}$
rapidly tends to $\frac{1}{e}$ with large $m$.  In summary, with
(\ref{Series_Bell_number}), we find that%
\[
\mathcal{B}_{m}=\sum_{j=0}^{m}\frac{j^{m}}{j!}\sum_{k=0}^{m-j}\frac{\left(
-1\right)  ^{k}}{k!}=\frac{1}{e}\sum_{j=0}^{\infty}\frac{j^{m}}{j!}.%
\]
which is quite remarkable.

\subsection{Stirling recursion lemma}\label{apx: stirling recurssion theorem}
\begin{lemma} \label{lemma: stirling recursion lemma}
    The solution to the recursive equation 
    \begin{equation}
    x_m = -\sum_{l=1}^{m-1} \mathcal{S}_m^{(l)} x_l
    \label{recursion_x_m}
    \end{equation}
    with initial conditions $x_1 = 1$ is
    \begin{align*}
        x_m = S_{m}^{(1)} = (-1)^{m-1} (m-1)!
    \end{align*}
    where \(S_{m}^{(1)}\) and \(\mathcal{S}_m^{(l)}\) are the Stirling numbers of the first and second kind.   
\end{lemma}

\begin{proof} 
Since $\mathcal{S}_{m}^{(k)}$ is the Stirling number of the second kind with $\mathcal{S}_{m}^{(m)}=1$ and $\mathcal{S}_{m}^{(1)}=1$, equation (\ref{recursion_x_m}) is equivalent to
\begin{equation}
x_m +\sum_{l=1}^{m-1} \mathcal{S}_m^{(l)} x_l = \mathcal{S}_{m}^{(m)}x_m + \sum_{l=1}^{m-1} \mathcal{S}_m^{(l)} x_l=
\sum_{l=1}^{m}\mathcal{S}_{m}^{(l)}x_{l}=0.\label{governing_eq}%
\end{equation}
Substitution of the initial condition $x_1 = 1$ and $\mathcal{S}_{m}^{(1)}=1$ into (\ref{governing_eq}) yields
\begin{equation} \label{governing_eq_simplified}
\sum_{l=2}^{m}\mathcal{S}_{m}^{(l)}x_{l}=-1.
\end{equation}
Multiplying both sides of (\ref{governing_eq_simplified}) by the Stirling number of the first kind $S_{q}^{(m)}$
\[
S_{q}^{(m)}\sum_{l=2}^{m}\mathcal{S}_{m}^{(l)}x_{l}=-S_{q}^{(m)}%
\]
and summing over the integers $m$ yields%
\[
\sum_{m=2}^{b}\sum_{l=2}^{m}S_{q}^{(m)}\mathcal{S}_{m}^{(l)}x_{l}=-\sum
_{m=2}^{b}S_{q}^{(m)}.%
\]
We reverse the summations%
\[
\sum_{m=2}^{b}\sum_{l=2}^{m}S_{q}^{(m)}\mathcal{S}_{m}^{(l)}x_{l}=\sum
_{l=2}^{b}\left(  \sum_{m=l}^{b}S_{q}^{(m)}\mathcal{S}_{m}^{(l)}\right)  x_{l}.%
\]
The second orthogonality formula of the Stirling numbers \cite[equation (13.14)]{book_combinatorics_app}%
\begin{equation}
\sum_{m=l}^{q}S_{q}^{(m)}\mathcal{S}_{m}^{(l)}=\delta_{lq}
\label{orthogonality_StirlingNumbers_2}%
\end{equation}
suggests to choose $b=q$ so that, for $q\geq2$,%
\[
-\sum_{m=2}^{q}S_{q}^{(m)}=\sum_{l=2}^{q}\left(  \sum_{m=l}^{q}S_{q}%
^{(m)}\mathcal{S}_{m}^{(l)}\right)  x_{l}=\sum_{l=2}^{q}\delta_{lq}x_{l}=x_{q}.%
\]
The generating function of the Stirling numbers of the first kind,
$m!\binom{x}{m}=\frac{\Gamma(x+1)}{\Gamma(x+1-m)}=\prod_{k=0}^{m-1}%
(x-k)=\sum_{k=0}^{m}S_{m}^{(k)}\;x^{k}$, shows that%
\[
\sum_{m=2}^{q}S_{q}^{(m)}=\sum_{k=0}^{q}S_{q}^{(k)}-S_{q}^{(1)}=-S_{q}^{(1)}.%
\]
Finally, $S_{q}^{(1)}=\left(  -1\right)  ^{q-1}\left(  q-1\right)  !$ and we
arrive, for $q\geq1$, at%
\[
x_{q}=S_{q}^{(1)}=\left(  -1\right)  ^{q-1}\left(  q-1\right)  !
\]
which proves Lemma \ref{lemma: stirling recursion lemma}. 
\end{proof}

\section{Algorithm for the probability distribution of RWIG graphs}\label{apx: thm1 algo}
\begin{algorithm}[H]
\caption{CalcSigma}\label{alg:calc sigma}
  \textbf{Input:}  $\textbf{s}_\mathcal{\mathcal{A}}[0]$ - initial states of walkers from set $\mathcal{A}$,\\ 
  \hspace*{1.2cm} $\textbf{P}_\mathcal{A}$ - Markov transition matrices of walkers from set $\mathcal{A}$,\\
  \hspace*{1.2cm} $k$ - discrete time.\\
  \textbf{Output:} $\sigma$ - the probability that walkers from set $\mathcal{A}$ are in the same state at discrete time $k$.
\begin{algorithmic}
\State $h \gets [1\ 1\ \cdots\ 1]$
\ForEach{$w \in \mathcal{A}$}
    \State $s_w[k] \gets s_w[0] P_w^k $
    \State $h \gets h \circ s_w[0]$
\EndFor
\State $\sigma \gets h \cdot [1\ 1\ \cdots\ 1]^T$
\State \Return $\sigma$
\end{algorithmic}
\end{algorithm}

\begin{algorithm}[H]
\caption{RWIG-pmf}\label{alg: thm1 algo}
\textbf{Input:}  $\mathcal{M}$ The set of $M$ random walkers, \\
\hspace*{1.2cm} $\textbf{s}_\mathcal{\mathcal{M}}[0]$ - initial states of \(M\) walkers,\\ 
  \hspace*{1.2cm} $\textbf{P}_\mathcal{M}$ - Markov transition matrices of \(M\) walkers,\\
  \hspace*{1.2cm} $k$ - discrete time,\\
  \hspace*{1.2cm} $g$ - contact graph realisation (partition of the walker set $\mathcal{M}$).\\
  \textbf{Output:} $p$ - the probability of $g$ at discrete time $k$. 
  
\begin{algorithmic}
\State $p \gets 0$
\State $\mathcal{P}^*_g \gets $  \textit{NonTrivialPartitions}($g$)

\ForEach {$\pi_g \in \mathcal{P}^*_g$}
    \State $\beta(\pi_g) \gets 1$
    \State $\sigma \gets 1$
    \ForEach{$\mathcal{C} \in \pi_g$}
        \State $\beta(\pi_g) \gets \beta(\pi_g) \times (-1)^{|\mathcal{C}|-1}(|\mathcal{C}|-1)!$
            \ForEach{$\mathcal{A} \in \mathcal{C}$}
                \State $\sigma \gets \sigma \times \text{\textit{CalcSigma}}(\mathcal{A}, \textbf{P}_\mathcal{A}, \mathbf{s}_\mathcal{A}[0], k)$
            \EndFor
    \EndFor
    \State $p \gets p + \sigma \beta(\pi_g)$
\EndFor
\State \Return $p$
\end{algorithmic}
\end{algorithm}

We use the Python programming language for our codebase. For the \textit{NonTrivialPartitions($g$)} function, we use a Python implementation of enumerating $l$-partitions using the \texttt{more-itertools} library. For all $l = \{1, 2, ..., |g| - 1\}$, we generate a list of $l$-partitions on $g$ using the \texttt{set\_partitions} function from the \texttt{more-itertools} library and concatenate the lists to create the set of non-trivial partitions $\mathcal{P}^*_g$.

\section{Computational complexity}\label{apx: time complexity}
For a random choice of discrete time $k$ and transition rate matrix $P$, we record the computational complexity of the entire probability distribution for an RWIG contact graph calculated using both Theorems \ref{thm: m cliques} and \ref{thm: main theorem}. The ratio $r = \frac{t_1}{t_2}$ between the execution time \(t_{1}\) using Theorem \ref{thm: m cliques} and the execution time \(t_{2}\) using Theorem \ref{thm: main theorem} is used to quantify the $r$-fold decrease in execution time when using Theorem \ref{thm: main theorem} instead of Theorem \ref{thm: m cliques}.

We tabulate in Figure \ref{fig:speedup} the speedup ratio \(r\) for various combinations of number of states $N$ and number of walkers $M$, and average the execution times over multiple iterations of calculating the probability distribution of the RWIG graph. We use Python's inbuilt \texttt{time.time()} function to record execution time. The reduction in time complexity quickly becomes relevant for small values of $M, N$, where for $M=N=7$ the 35-fold decrease in execution time speeds up the probability distribution calculation from 18.29 seconds to 0.52 seconds.

\begin{figure}[!t]
    \centering
    \includegraphics[width = \linewidth]{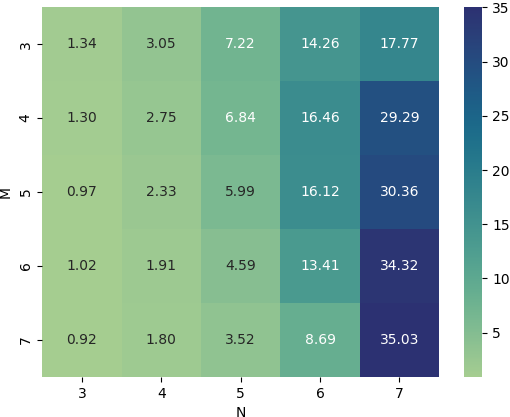}
    \caption{Speedup ratio heatmap for various combinations of $M$ and $N$}
    \label{fig:speedup}
\end{figure}

\section{Sigma products weights}\label{apx: lemma 1 2 proof}
\subsection{$1$-partition weight $\beta_m(\pi_1)$}
\setcounter{lemma}{0}
\begin{lemma}
    Let $\pi_1 = \{\mathcal{M}\}$ be the $1$-partition on the walker set $\mathcal{M}$. The number $\beta_m(\pi_1)$  of sigma product terms $\sigma_\mathcal{M}[k]$ in the sigma expansion formula (\ref{eq:sigma expansion}) for the probability of an $m$-clique contact graph depends only on the number of cliques $m$ as
    \begin{equation}
        \beta_m(\pi_1) = (-1)^{(m-1)}(m-1)!
    \end{equation}
\end{lemma}
\begin{proof}
    In (\ref{eq:2cliques closedform 1}), we show that any 2-clique contact graph probability has the same $\sigma_\mathcal{M}[k]$ weight $\beta_2(\pi_1) = -1$, because
    \begin{align}\label{eq:2cliques closedform}
\Pr[G_k = \{\mathcal{A}_1, \mathcal{A}_2\} ] &= \sigma_{\mathcal{A}_1}[k]\sigma_{\mathcal{A}_2}[k] - \sigma_\mathcal{M}[k]\nonumber\\
&= \beta_2(\{\{\mathcal{A}_1\}, \{\mathcal{A}_2\}\}) \sigma_{\mathcal{A}_1}[k]\sigma_{\mathcal{A}_2}[k]\nonumber\\
&\ \ \ + \beta_2(\pi_1) \sigma_\mathcal{M}[k].
    \end{align}
    Similarly, equation (\ref{eq: 3cliques closedform}) indicates that $\beta_3(\pi_1) = 2$. Naturally, $\beta_1(\pi_1) = 1$, because the probability $\Pr[G_k = \{\mathcal{M}\} ] = \Pr[\text{all walkers in same Markov state at time $k$}] := \sigma_\mathcal{M}[k]$. We prove the general case for $\beta_m(\pi_1)$ by induction.
    
    Consider all non-singleton \textit{l}-partitions on a \textit{m}-clique contact graph $g =  \{\mathcal{A}_1, \mathcal{A}_2, ..., \mathcal{A}_m\}$, where $l < m$. Splitting the summation over all non-singleton partitions $\mathcal{P}^{*}_{g}$ in (\ref{eq: kcliques recursive}) into \textit{l}-partitions yields:
    \begin{equation}\label{eq:l-partition expansion}
        \Pr[G_k = g ] = \prod_{i=1}^m \sigma_{\mathcal{A}_i}[k] - \sum_{l=1}^{m-1} \sum_{\pi \in \mathcal{P}_{g}(l)} \Pr[G_k = g(\pi) ],
    \end{equation}
    where $\mathcal{P}_{g}(l)$ is the set of \textit{l}-partitions on $g$. Consequently, each amassed-clique contact graph $g(\pi)$ is an \textit{l}-clique contact graph. 
    
    Assume that the number of terms $\sigma_{\mathcal{M}}$ in the sigma expansion (\ref{eq:sigma expansion}) of any \textit{l}-clique contact graph is $\beta_l(\pi_1), \text{for each positive integer}\ l < m$. Appendix \ref{apx: bell,stirling definitions} shows that the Stirling number of the second kind $\mathcal{S}_m^{(l)}$ is equal to the number of $l$-partitions on a set with $m$. Therefore, the number of elements in  the set $\mathcal{P}_{g}(l)$ is $|\mathcal{P}_{g}(l)| = \mathcal{S}_m^{(l)}$. We perform a sigma expansion for each amassed contact graph $g(\pi)$ in (\ref{eq:l-partition expansion}). Since we are interested in equating $\sigma_\mathcal{M}[k]$ terms  in (\ref{eq:sigma expansion}) and the sigma expansions in (\ref{eq:l-partition expansion}), we omit all terms except for $\sigma_\mathcal{M}[k]$,
    \begin{align}\label{eq: expanding}
        \Pr[G_k = g] &= \prod_{i=1}^m \sigma_{\mathcal{A}_i}[k] - \sum_{l=1}^{m-1} \sum_{\pi \in \mathcal{P}_{g}(l)} \Pr[G_k = g(\pi) ] \nonumber\\
        &= -\sum_{l=1}^{m-1} \mathcal{S}_m^{(l)} \beta_l(\pi_1) \sigma_\mathcal{M}[k] + ...
    \end{align}
    Equating $\sigma_\mathcal{M}[k]$ terms in (\ref{eq:sigma expansion}) and (\ref{eq: expanding}) yields the recursion
    \begin{equation}\label{eq: recursive stirling}
        \beta_m(\pi_1) = -\sum_{l=1}^{m-1} \mathcal{S}_m^{(l)} \beta_l(\pi_1).
    \end{equation}
    Lemma \ref{lemma: stirling recursion lemma} in Appendix \ref{apx: stirling recurssion theorem} gives the solution of recursion (\ref{eq: recursive stirling}) with initial condition $\beta_1(\pi_1)=1$ as
    \(
        \beta_m(\pi_1) = (-1)^{(m-1)}(m-1)!
    \), which demonstates Lemma \ref{lemma: trivial partition weight}.
\end{proof}

\subsection{General $q$-partition weight $\beta_m(\pi_q)$}
\begin{lemma}
    Let $g$ be a  $m$-clique contact graph. Let $\pi_q = \{\mathcal{C}_1, ..., \mathcal{C}_q\}$ be a $q$-partition on $g$, with $q < m$. Let the cardinality of each cell $\mathcal{C}_i$ be $c_i$. Let the number of sigma product terms $\prod_{i=1}^q \sigma_{\mathcal{C}_i}[k]$ in the sigma expansion formula of $g$ be $\beta_m(\pi_q)$. Then
    \begin{equation}
        \beta_m(\pi_q) = \prod_{i=1}^q (-1)^{c_i-1} (c_i-1)!
    \end{equation}
\end{lemma}
\begin{proof}
    Consider $q$ distinct sets of walkers $\{\mathcal{C}_{i}\}_{i=1}^q$. Each set $\mathcal{C}_i$ forms a $c_i$-clique contact graph $g(i) = \{\mathcal{A}_1(i),... \mathcal{A}_{c_i}(i)\}$. For each set $\mathcal{C}_i$, the sigma expansion formula (\ref{eq:sigma expansion}) for the probability $\Pr[G_k(i) = g(i) | \textbf{P}_{\mathcal{C}_i}, \mathbf{s}_{\mathcal{C}_i}[0], k]$ contains $\beta_{c_i}(\{\mathcal{C}_i\}) = \beta_{c_i}(\pi_1)= (-1)^{c_i-1} (c_i-1)!$  terms $\sigma_{\mathcal{C}_i}[k]$ by Lemma \ref{lemma: trivial partition weight}. 
    
    Consider now the union of all walker sets $\mathcal{M} = \cup_{i=1}^q \mathcal{C}_i$. By definition (\ref{eq:sigmaA}), the probability that all walkers from $\mathcal{C}_i$ are in the same state equals $\sigma_{\mathcal{C}_i}[k]$. Since walkers are all independent, the contact graphs $g(i)$, formed by the walker sets $\mathcal{C}_i$, for all $i \in \{1, 2, ..., q\}$ are also independent. Therefore, the probability of the contact graph realisation $g$ generated by all walkers in $\mathcal{M}$ is
    \begin{align}\label{eq: general expansion in general lemma}
        \Pr[G_k = g ] = \prod_{i=1}^q \Pr[G_k(i) = g(i) | \textbf{P}_{\mathcal{C}_i}, \mathbf{s}_{\mathcal{C}_i}[0], k].
    \end{align}
    Performing a sigma expansion for each contact graph probability term $\Pr[G_k(i) = g(i) |  \textbf{P}_{\mathcal{C}_i}, \mathbf{s}_{\mathcal{C}_i}[0], k]$ in (\ref{eq: general expansion in general lemma}) and omitting all terms, which are not $\prod_{i=1}^q \sigma_{\mathcal{M}_i}[k]$, yields
    \begin{align*}
        \Pr[G_k = g ] = \prod_{i=1}^q \beta_{c_i}(\{\mathcal{C}_i\}) \sigma_{\mathcal{C}_i}[k] + ...
    \end{align*}
    But $\beta_{c_i}(\{\mathcal{C}_i\}) = (-1)^{c_i-1} (c_i-1)!$ and thus:
    \begin{align}\label{eq:general pi' partition weight}
        \Pr[G_k = g ] = \left(\prod_{i=1}^q (-1)^{c_i-1} (c_i-1)!\right)\left( \prod_{i=1}^q \sigma_{\mathcal{C}_i}[k] \right) + ...
    \end{align}
    Since the sets $\{\mathcal{C}_{i}\}_{i=1}^q$ represent the cells of a $q$-partition on the walker set $\mathcal{M}$, equation (\ref{eq:general pi' partition weight}) states that the weight $\beta_{|\mathcal{M}|}(\pi_q)$ of the sigma product term $\prod_{i=1}^q \sigma_{\mathcal{C}_i}[k]$ associated with a general $q$-partition $\pi_q = \{\mathcal{C}_1, ..., \mathcal{C}_q\}$ on $g$ is $\prod_{i=1}^q (-1)^{c_i-1} (c_i-1)!$.
\end{proof}

\section{4-clique contact graph probability}\label{apx: 4 clique graph proba}
Let $g = \{\mathcal{A}_1, \mathcal{A}_2, \mathcal{A}_3, \mathcal{A}_4\}$. The formula for the probability \(\Pr[G_k = g]\) is

\begin{align}
        \Pr[G_k = g] = & \sigma_{\mathcal{A}_1}[k] \sigma_{\mathcal{A}_2}[k] \sigma_{\mathcal{A}_3}[k] \sigma_{\mathcal{A}_4}[k] \nonumber \\
    &- \Pr[G_k = \{\mathcal{A}_{12}, \mathcal{A}_3, \mathcal{A}_4\}] \nonumber \\
    &- \Pr[G_k = \{\mathcal{A}_{13}, \mathcal{A}_2, \mathcal{A}_4\}] \nonumber \\
    &- \Pr[G_k = \{\mathcal{A}_{14}, \mathcal{A}_2, \mathcal{A}_3\}] \nonumber \\
    &- \Pr[G_k = \{\mathcal{A}_{23}, \mathcal{A}_1, \mathcal{A}_4\}] \nonumber \\
    &- \Pr[G_k = \{\mathcal{A}_{24}, \mathcal{A}_1, \mathcal{A}_3\}] \nonumber \\
    &- \Pr[G_k = \{\mathcal{A}_{34}, \mathcal{A}_1, \mathcal{A}_2\}] \nonumber \\
    \nonumber\\
    &- \Pr[G_k = \{\mathcal{A}_{12}, \mathcal{A}_{34}\}] \nonumber \\
    &- \Pr[G_k = \{\mathcal{A}_{13}, \mathcal{A}_{24} \}] \nonumber \\
    &- \Pr[G_k = \{\mathcal{A}_{14}, \mathcal{A}_{23}\}] \nonumber \\
    \nonumber\\
    &- \Pr[G_k = \{\mathcal{A}_{123}, \mathcal{A}_4\}] \nonumber \\
    &- \Pr[G_k = \{\mathcal{A}_{124}, \mathcal{A}_3\}] \nonumber \\
    &- \Pr[G_k = \{\mathcal{A}_{134}, \mathcal{A}_2\}] \nonumber \\
    &- \Pr[G_k = \{\mathcal{A}_{234}, \mathcal{A}_1\}] \nonumber \\
    \nonumber\\
    &- \Pr[G_k = \{\mathcal{M}\}].\nonumber
    \end{align}
Introducing the sigma terms definition (\ref{eq:sigmaA}) yields 
\begin{align}
    &\Pr[G_k = g] = \sigma_{\mathcal{A}_1}[k] \sigma_{\mathcal{A}_2}[k] \sigma_{\mathcal{A}_3}[k] \sigma_{\mathcal{A}_4}[k] \nonumber\\
    &- \sigma_{\mathcal{A}_{12}}[k] \sigma_{\mathcal{A}_3}[k] \sigma_{\mathcal{A}_4}[k] - \sigma_{\mathcal{A}_{13}}[k] \sigma_{\mathcal{A}_2}[k] \sigma_{\mathcal{A}_4}[k]\nonumber\\
    &-\sigma_{\mathcal{A}_{14}}[k] \sigma_{\mathcal{A}_2}[k] \sigma_{\mathcal{A}_3}[k] - \sigma_{\mathcal{A}_{23}}[k] \sigma_{\mathcal{A}_1}[k] \sigma_{\mathcal{A}_4}[k]\nonumber\\
    &- \sigma_{\mathcal{A}_{24}}[k] \sigma_{\mathcal{A}_1}[k] \sigma_{\mathcal{A}_3}[k] - \sigma_{\mathcal{A}_{34}}[k] \sigma_{\mathcal{A}_1}[k] \sigma_{\mathcal{A}_2}[k] \nonumber\\
    &+2(\sigma_{\mathcal{A}_{123}}[k]\sigma_{\mathcal{A}_4}[k] + \sigma_{\mathcal{A}_{124}}[k]\sigma_{\mathcal{A}_3}[k] + \sigma_{\mathcal{A}_{134}}[k]\sigma_{\mathcal{A}_2}[k]\nonumber\\
    & + \sigma_{\mathcal{A}_{234}}[k]\sigma_{\mathcal{A}_1}[k]) \nonumber\\
    &+ \sigma_{\mathcal{A}_{12}}[k]\sigma_{\mathcal{A}_{34}}[k] + \sigma_{\mathcal{A}_{13}}[k]\sigma_{\mathcal{A}_{24}}[k] + \sigma_{\mathcal{A}_{14}}[k]\sigma_{\mathcal{A}_{23}}[k] \nonumber\\
    &- 6 \sigma_{\mathcal{M}}[k].\nonumber
\end{align}

\InputIfFileExists{bibliography.bbl}{}{}


\begin{thebibliography}{1}
\bibitem{covid1} 
Yap K. Y.-L. \& Xie Q. "Personalizing symptom monitoring and contact tracing efforts through a COVID-19 web-app",
\newblock {\em Infect Dis Poverty} 9, 93, 2020.

\bibitem{covid2} 
Cencetti G., Santin G., Longa A., Pigani E., Barrat A., Cattuto C., Lehmann S., Salathé M. \&  Lepri B. "Digital proximity tracing on empirical contact networks for pandemic control",
\newblock {\em Nat Commun}
\newblock 12, 1655, 2021. 

\bibitem{covid3}
World Health Organization.
\newblock {\em Digital Tools for COVID-19 Contact Tracing},
\newblock {Geneva, Switzerland, 2020. \url{https://who.int/publications/i/item/WHO-2019-nCoV-Contact\_Tracing-Tools\_Annex-2020.1}}.

\bibitem{Iribarren2009}
Iribarren J.L. \& Moro E. "Impact of Human Activity Patterns on the Dynamics of Information Diffusion",
\newblock{\em Phys. Rev. Lett.} 103. 038702, 2009.

\bibitem{Karsai2013}
Karsai M., Kivel\"a M., Pan R.K., Kaski K., Kert\'esz J., Barab\'asi A.-L. \& Saram\"aki J. "Small but slow world: How network topology and burstiness slow down spreading",
\newblock {\em Phys Rev. E}
\newblock 83, 2013.

\bibitem{Takaguchi2013}
Takaguchi P. "Bursty communication patterns facilitate spreading in a threshold-based epidemic dynamics",
\newblock{\em PLOS ONE}
\newblock{8(7), 1-5, 2013.}

\bibitem{Perotti2014} Perotti J.I., Jo H.-H., Holme P. \& Saramäki J. "Temporal network sparsity and the slowing down of spreading",
\newblock{\em arXiv}:1411.5553,2014.

\bibitem{Masuda2014}
Masuda N., Klemm K., \& Eguíluz V. "Temporal Networks: Slowing Down Diffusion by Long Lasting Interactions",
\newblock{\em Phys. Rev. Lett.}
\newblock{111, 188701, 2013.}

\bibitem{Morris1995}
Morris M. \& Kretzschmar M. "Concurrent partnerships and transmission dynamics in networks",
\newblock{\em Soc. Netw.}
\newblock{17(3), 299-318, 1995.}

\bibitem{Lentz2013}
Lentz H.H.K., Selhorst T. \& Sokolov I.M. "Unfolding Accessibility Provides a Macroscopic Approach to Temporal Networks",
\newblock {\em Phys. Rev. Lett.}
\newblock 110, 118701, 2013.

\bibitem{Pan2011}
Pan R.K. \& Saramäki J. "Path lengths, correlations, and centrality in temporal networks",
\newblock {\em Phys. Rev. E}
\newblock 84, 016105, 2011.

\bibitem{Holme2012}
Holme P. \& Saramäki J. "Temporal networks",
\newblock {\em Phys. Rep.}
\newblock 519(3), 97-125, 2012.

\bibitem{Pfitzner2013}
Pfitzner R., Scholtes I., Garas A., Tessone C.J. \& Schweitzer F. "Betweenness Preference: Quantifying Correlations in the Topological Dynamics of Temporal Networks",
\newblock {\em Phys. Rev. Lett.} 110, 198701, 2013.

\bibitem{Scholtes2014}
Scholtes I., Wider N., Pfitzner R., Garas A., Tessone C.J. \& Schweitzer F. "Causality-driven slow-down and speed-up of diffusion in non-Markovian temporal networks",
\newblock{\em Nat Commun}
\newblock 5, 5024, 2014.

\bibitem{Rosvall2014} 
Rosvall M., Esquivel A.V., Lancichinetti A., West J.D. \& Lambiotte R. "Memory in network flows and its effects on spreading dynamics and community detection",
\newblock {\em Nat Commun}
\newblock 5, 4630, 2014.

\bibitem{Lentz2016}
Lentz H.H.K., Koher A., Hövel P., Gethmann J., Sauter-Louis C., Selhorst T. \& Conraths F.J. "Disease Spread through Animal Movements: A Static and Temporal Network Analysis of Pig Trade in Germany"
\newblock {\em PLOS ONE}
\newblock 11(5), e0155196, 2016.

\bibitem{Bois2015}
Bois F. \& Gayraud G. "Probabilistic generation of random networks taking into account information on motifs occurrence",
\newblock {\em J Comput Biol}
\newblock {22(1): 25-36, 2015.}

\bibitem{Scholtes2017} Scholtes I.  "When is a Network a Network? Multi-Order Graphical Model Selection in Pathways and Temporal Networks", in\newblock {\em Proc. 23rd ACM SIGKDD International Conference on Knowledge Discovery \& Data Mining}, KDD ‘17, 1037–1046, 2017.


\bibitem{Zhou2020}Zhou D., Zheng L., Han J. \& He J. "A Data-Driven Graph Generative Model for Temporal Interaction Networks", in \textit{Proc. 26th ACM SIGKDD International Conference on Knowledge Discovery \& Data Mining}, KDD ‘20, 401–411, 2020.


\bibitem{Zeno2021}Zeno G., La Fond T. \& Neville J. "Dymond: Dynamic motif-nodes network generative model", \textit{Proc. Web Conf. 2021}, 718–729, 2021.

\bibitem{Longa2024}Longa A., Cencetti G., Lehmann S., Passerini A. \& Lepri B. "Generating fine-grained surrogate temporal networks",  \textit{Commun Phys} 7, 22, 2024.

\bibitem{Lambiotte2019} Lambiotte R., Rosvall M. \& Scholtes I. "From networks to optimal higher-order models of complex systems", \textit{Nat. Phys.} 15, 313–320, 2019.

\bibitem{Shvydun_Van Mieghem2024} Shvydun S. \& Van Mieghem P. "System Identification for Temporal Networks", \textit{IEEE Trans. Netw. Sci. Eng.} 11(2), 1885-1895, 2024.

\bibitem{Barbosa2018}Barbosa H., Barthelemy M., Ghoshal G., James C., Lenormand M., Louail T., Menezes R., Ramasco J., Simini F. \& Tomasini M. "Human mobility: Models and applications", \textit{Phys. Rep.} 734, 1-74, 2018.

\bibitem{Chang2022}Chang B., Yang L., Sensi M., Achterberg M., Wang F., Rinaldi M. \& Van Mieghem P. "Markov Modulated Process to Model Human Mobility", \textit{Stud. Comput. Intell.} 1072, 607–618, 2022.

\bibitem{Panisson2012}Panisson A., Barrat A., Cattuto C., Van den Broeck W., Ruffo G. \& Schifanella R. "On the dynamics of human proximity for data diffusion in ad-hoc networks", \textit{Ad Hoc Netw.} 10(8), 1532-1543, 2012.

\bibitem{Mauro2022}Mauro G., Luca M., Longa A., Lepri B. \& Pappalardo L. "Generating mobility networks with generative adversarial networks", \textit{EPJ Data Sci.} 11(58), 2022.

\bibitem{Kui2018}
Kui X., Samanta A., Zhu X., Zhang S., Li Y. \& Hui P. "Energy-Aware Temporal Reachability Graphs for Time-Varying Mobile Opportunistic Networks",
\newblock {\em IEEE Trans. Veh. Technol.} 67(10), 2018.

\bibitem{Huyn2022}
Huynh N. \& Barthelemy J. "A comparative study of topological analysis and temporal network analysis of a public transport system",
\newblock {\em Int. J. Transp. Sci. Technol.}
\newblock 11(2), 392-405, 2022.

\bibitem{meanencountertimes} 
Riascos A.P. \& Sanders D.P. "Mean encounter times for multiple random walkers on networks",
\newblock {\em Phys. Rev. E}
\newblock 103, 042312, 2021.

\bibitem{markovian random walk model of epidemic spreading} 
Bestehorn M., Riascos A., Michelitsch T. \& Collet B. "A Markovian random walk model of epidemic spreading",
\newblock {\em Continuum Mech. Thermodyn}
\newblock 33, 1207–1221, 2021.

\bibitem{comprehensiveRWstudy} 
Masuda N., Porter M. \& Lambiotte R. "Random walks and diffusion on networks",
\newblock {\em Phys. Rep.} 
\newblock 716, 1-58, 2017.


\bibitem{PVM_graphspectra_second_edition}
Van Mieghem P. 
\newblock {\em Graph Spectra for Complex Networks}.
\newblock Second edition, Cambridge University Press, 2023.

\bibitem{perfanalysis}Van Mieghem P. \emph{Performance Analysis of Complex Networks and Systems}. Cambridge University Press, 2014.

\bibitem{book_abramowitz} Abramowitz M. \& Stegun I.A. \emph{Handbook of Mathematical Functions}. Dover Publications, New York, 1968.

\bibitem{book_combinatorics}Van Lint J.H. \& Wilson R.M. \emph{A Course in Combinatorics}. Second edition. Cambridge University Press, 2001.

\bibitem{co-location}
Génois M. \& Barrat A. "Can co-location be used as a proxy for face-to-face contacts?", \textit{EPJ Data Sci.} 7(11), 2018.

\end{thebibliography}

\begin{thebibliography}{5}
\bibitem{PVM_graphspectra_second_edition_app}
Van Mieghem P. 
\newblock {\em Graph Spectra for Complex Networks}.
\newblock Second edition, Cambridge University Press, 2023.

\bibitem{book_combinatorics_app}Van Lint J.H. \& Wilson R.M. \emph{A Course in Combinatorics}. Second edition. Cambridge University Press, 2001.

\bibitem{book_comtet} Comtet L. \emph{Advanced Combinatorics}, revised and enlarged edition, D. Riedel Publishing Company, 1974.

\bibitem{PVM_charcoef}
Van~Mieghem P.
\newblock {\em Characteristic Coefficients of a Complex Function}.
\newblock unpublished, 1993-2024.

\bibitem{book_abramowitz_app} Abramowitz M. \& Stegun I.A. \emph{Handbook of Mathematical Functions}. Dover Publications, New York, 1968.

\end{thebibliography}
\end{document}